\documentclass[pdflatex,sn-mathphys-num]{sn-jnl}
\usepackage{multirow,soul,graphicx}%
\usepackage{amsmath,amssymb,amsfonts,a4wide}%
\usepackage{amsthm}%
\usepackage{mathrsfs}%
\usepackage[title]{appendix}%
\usepackage{xcolor}%
\usepackage{textcomp}%
\usepackage{manyfoot}%
\usepackage{booktabs}%
\usepackage{algorithm}%
\usepackage{algorithmicx}%
\usepackage{algpseudocode}%
\usepackage{listings}%

\theoremstyle{thmstyleone}%
\newtheorem{theorem}{Theorem}
\newtheorem{proposition}[theorem]{Proposition}%

\theoremstyle{thmstyletwo}%
\newtheorem{remark}{Remark}%
\newtheorem{lemma}{Lemma}%

\theoremstyle{thmstylethree}%

\begin{document}

\title[Two-strain epidemic model]{A new oscillatory regime in two-strain epidemic models with partial cross-immunity}

\author*[1]{\fnm{Nir} \sur{Gavish}}\email{ngavish@technion.ac.il}
\affil*[1]{\orgdiv{Faculty of Mathematics}, \orgname{Technion - IIT}, \orgaddress{\city{Haifa}, \postcode{32000}, \country{Israel}}}

\abstract{
Infectious diseases often involve multiple strains that interact through the immune response generated after an infection. This study investigates the conditions under which a two-strain epidemic model with partial cross-immunity can lead to self-sustained oscillations, and reveals a new oscillatory regime in these models. Contrary to previous findings, which suggested that strong cross-immunity and significant asymmetry between strains are necessary for oscillations, our results demonstrate that sustained oscillations can occur even with weak cross-immunity and weak asymmetry. Using asymptotic methods, we provide a detailed mathematical analysis showing that the steady state of coexistence becomes unstable along specific curves in the parameter space, leading to oscillatory solutions for any value of the basic reproduction number greater than one. Numerical simulations support our theoretical findings, highlighting an unexpected oscillatory region in the parameter domain. These results challenge the current understanding of oscillatory dynamics in multi-strain epidemiological models, point to an oversight in previous studies, and suggest broader conditions under which such dynamics can arise.
}


\maketitle

\section{Introduction}\label{sec:intro}
Infectious diseases often consist of multiple strains with examples including seasonal influenza, Haemophilus influenza, Streptococcus pneumonia, human immunodeficiency virus (HIV), tuberculosis, and Dengue virus; see~\cite{martcheva2015introduction,wormser2008modeling} and references within. These strains can interact indirectly through the immune response generated after an infection. Specifically, with {\em partial cross-immunity}, the immune response to an infection can provide some protection against subsequent infections with related strains. Here, $\sigma_{ij}$ represents the relative susceptibility to strain $j$ for an individual previously infected with strain $i$. If $\sigma_{ij}=0$, there is complete cross-immunity; $0<\sigma_{ij}<1$ indicates partial cross-immunity; and $\sigma_{ij}=1$ means no cross-immunity.

Epidemic models with multiple strains and partial cross-immunity can exhibit oscillatory dynamics. Previous studies have shown this behavior with as few as three interacting strains~\cite{andreasen1997dynamics,gomes2002dynamics,dawes2002onset,lin1999dynamics} or when factoring in changes in infectivity, quarantine, or age structure~\cite{ferguson1999effect,gupta1998chaos,nuno2005dynamics,nuno2008mathematical,thieme2007pathogen,kuddus2022analysis,castillo1989epidemiological}. However, the above works leave open the question whether oscillatory dynamics can arise in a simpler two-strain model solely due to partial cross-immunity, without additional factors.

Chung and Lui~\cite{chung2016dynamics} provide an answer to this question by demonstrating the emergence of self-oscillations in a two-strain influenza model.  The picture arising from their work is that two conditions are necessary for self-sustained oscillations:
1) {\em Strong cross-immunity:} The cross-immunity induced by at least one of the strains must be sufficiently strong.  This condition is intuitive since it is known that a multi-strain SIR system does not cause oscillations in the absence of cross-immunity~($\sigma_{12}=\sigma_{21}=1$).  Thus, it can be expected that systems with weak cross-immunity~($\sigma_{ij}\approx 1$) will also not give rise to oscillations.  This condition also arises in other works~\cite{andreasen1997dynamics,gomes2002dynamics,dawes2002onset,lin1999dynamics}.  2) {\em Strong asymmetry:} Adequate separation in cross-immunity between the two strains is essential. This condition comes from a proof suggesting that when there is symmetry or weak asymmetry, the steady state of co-existence is stable~\cite{chung2016dynamics}. Nevertheless, we will reveal a caveat in this proof.

Contrary to the above, our study shows that sustained oscillations can also occur with weak cross-immunity and weak asymmetry. In particular, we provide an example of sustained oscillations where one strain offers no cross-immunity ($\sigma_{12}=1$) and the other offers only weak cross-immunity ($\sigma_{21}=0.95$). This finding challenges current understanding and suggests that oscillatory dynamics can arise in a different range of conditions than previously thought.  

Our first key result, Theorem~\ref{theo:main_result}, shows that when one strain does not provide effective cross-immunity ($\sigma_{12}\approx1$), and the other strain provides partial cross-immunity ($0<\sigma_{21}<1$), there is a unique steady-state of coexistence that is unstable for a certain region in the parameter space. 
The discrepancy between our results and Theorem 1.2 in~\cite{chung2016dynamics} arises because the latter assumes that an approximation of the relevant eigenvalues is uniformly valid, while we prove that the approximation becomes invalid in certain regions of the parameter space, specifically along a curve we characterize.  Moreover, we directly associate the instabilities of the coexistence steady state with the cases in which the approximation is invalid.

As in~\cite{castillo1989epidemiological,chung2016dynamics}, our analysis takes advantage of the separation of time scales that arises naturally in many scenarios, especially for common pathogens in humans.  For example, in Influenza, Varicella, COVID-19, Rubeola, and others~\cite{CDCyellowbook2024}, the characteristic infectious period is on a scale of a week compared to a human life span of 70-75 years.   In terms of the model parameters, these examples correspond to a case in which the mortality rate is significantly lower than the recovery rates.  Using asymptotic methods, we are able to approximate the relevant quantities with the accuracy required to determine the stability of the coexistence steady state.

Nonetheless, our second main finding stands without approximations and applies to all positive mortality rates. We demonstrate the existence and uniqueness of a coexisting steady state within a specific parameter plane region, with its values derived from solving a cubic equation. This outcome serves as a robust foundation for further stability analysis.

In summary, our work expands the understanding of self-sustained oscillations in epidemiological systems, revealing an oversight in previous works and highlighting an unexpected oscillatory region in the parameter domain.

The paper is organized as follows. In Section~\ref{sec:mathModel}, we present the mathematical model and review relevant known results. In Section~\ref{sec:phiCE}, we consider the nonlinear algebraic system of eight equations for the compartment sizes of the coexistence steady-state, present a reduction of it to a cubic equation, and prove the existence of a unique feasible solution.  In Section~\ref{sec:CE_perturbation}, we focus on cases in which the mortality rate is significantly lower than the recovery rates.  Using asymptotic methods, we characterize the maximal parameter domain
in which there exists a feasible steady-state solution of coexistence and approximate its value.  Section~\ref{sec:stability_CE} is devoted to stability analysis of the coexistence steady-state, and in particular to the proof of Theorem~\ref{theo:main_result}.  Numerical examples are provided in Section~\ref{sec:numerics}.  In Section~\ref{sec:uniformExpansion}, we extend the stability analysis and provide a uniform approximation of the relevant eigenvalues in the full parameter space.  Section~\ref{sec:conclusions} provides a summary and discussion of our results.

\section{Model}\label{sec:mathModel}
We consider the transmission model presented in~\cite{chung2016dynamics,castillo1989epidemiological,martcheva2015introduction} for the dynamics of two strains with possible partial cross-immunity, i.e., when infection by one strain possibly provides reduced susceptibility to infection by the other strain.  For readability, in this section we present the mathematical model and review the related results~\cite{chung2016dynamics,castillo1989epidemiological,martcheva2015introduction}.

\begin{figure}[ht!]
\begin{center}
\includegraphics[width=0.75\textwidth]{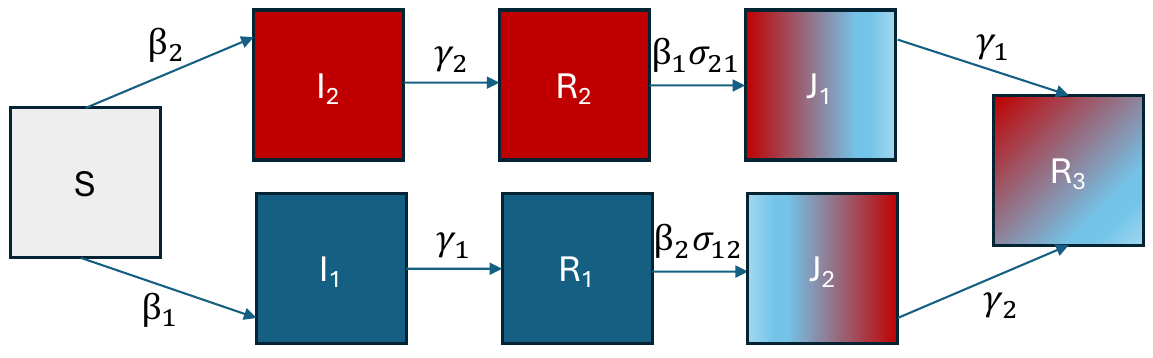}
\caption{Schematic diagram of disease dynamics when the host is exposed to two co-circulating strains. Susceptibles (S) may be infected with strain~$i=1$ or~$i=2$ ($I_i$, primary infection).  Those recovered from strain~$i$ ($R_i$, as a result of primary infection) are immune to reinfections by strain~$i$ but are possibly susceptible to infections by strain~$j\ne i$ ($J_j$, secondary infection).  Here,~$\sigma_{ij}$ is the relative susceptibility to strain~$j$ for an individual previously infected with and recovered from strain~$i$ ($i\ne j$), so that~$0<\sigma_{ij}<1$ corresponds to partial cross-immunity and~$\sigma_{ij}=1$ corresponds to a case with no cross-immunity.}
\label{fig:diagram}
\end{center}
\end{figure} The population is divided into eight different relative compartments: Susceptible (S), those infected with strain~$i$ ($I_i$, primary infection), those recovered from strain~$i$ ($R_i$, as a result of primary infection), those infected with strain~$i$ after they had recovered from strain~$j\ne i$ ($J_i$, secondary infection) and those recovered from both strains ($R_3$), see diagram in Figure~\ref{fig:diagram}. The population is assumed to mix randomly.  The model is given by~\cite{chung2016dynamics}
\begin{subequations} \label{eq:model}
\begin{align}
        \frac{\text{d}S}{\text{d}t} &=\mu-\sum_{i=1}^{2}\beta_i(I_i+J_i)S-\mu S,\label{eq:eqS}\\[1ex]
         \frac{\text{d}I_1}{\text{d}t}  &= \beta_1S(I_1+J_1)-(\mu+\gamma_1)I_1,\label{eq:eqI1}\\[1ex]
         \frac{\text{d}I_2}{\text{d}t}  &= \beta_2S(I_2+J_2)-(\mu+\gamma_2)I_2,\label{eq:eqI2}\\[1ex]
         \frac{\text{d}J_1}{\text{d}t}  &= \beta_1\sigma_{21}R_2(I_1+J_1)-(\mu+\gamma_1)J_1,\label{eq:eqJ1}\\[1ex]
         \frac{\text{d}J_2}{\text{d}t}  &= \beta_2\sigma_{12}R_1(I_2+J_2)-(\mu+\gamma_2)J_2,\label{eq:eqJ2}\\[1ex]
         \frac{\text{d}{R_1}}{\text{d}t} &= \gamma_1I_1-\beta_2\sigma_{12}(I_2+J_2)R_1-\mu R_1,\label{eq:eqR1}\\[1ex]
         \frac{\text{d}{R_2}}{\text{d}t} &= \gamma_2I_2-\beta_1\sigma_{21}(I_1+J_1)R_2-\mu R_2,\label{eq:eqR2}\\[1ex]
         \frac{\text{d}{R_3}}{\text{d}t} &=\gamma_1J_1+\gamma_2J_2-\mu R_3,\nonumber
\end{align}
\end{subequations}
were~$\mu$ is the birth rate, as well as the mortality rate,~$\beta_i$ denotes the transmission coefficient for strain~$i$, and~$\gamma_i$ denotes the recovery rate from strain~$i$. Additionally,~$\sigma_{ij}$ is the relative susceptibility to strain~$j$ for an individual previously infected with and recovered from strain~$i$ ($i\ne j$), so that~$0<\sigma_{ij}<1$ corresponds to reduced susceptibility (partial cross-immunity) and~$\sigma_{ij}=1$ corresponds to neutral susceptibility (no cross-immunity). 

We consider initial conditions which satisfy
\begin{equation}
\begin{split}
&S(0)\ge0, I_i(0)\ge0, J_i(0)\ge0, R_i(0)\ge0, R_3(0)\ge0\quad i=1,2,\\
&S(0)+I_1(0)+I_2(0)+J_1(0)+J_2(0)+R_1(0)+R_2(0)+R_3(0)=1,
\end{split}
\end{equation}
and, thus, ensure that for all~$t>0$,
\begin{equation}\label{eq:boundOnVariables}
\begin{split}
&S(t)\ge 0, I_i(t)\ge0, J_i(t)\ge0, R_i(t)\ge0, R_3(0)\ge0\quad i=1,2,\\
&S(t)+I_1(t)+I_2(t)+J_1(t)+J_2(t)+R_1(t)+R_2(t)+R_3(t)\equiv 1.
\end{split}
\end{equation}
In what follows, we use the algebraic relation~\eqref{eq:boundOnVariables} to express~$R_3(t)$ in terms of the other compartment sizes, and consider the model as a system of seven equations~(\ref{eq:eqS}-\ref{eq:eqR2}).
\subsection{Assumptions on model parameters}
We consider the model~\eqref{eq:model} in the case of partial or neutral cross-immunity,~$0<\sigma_{ij}\le 1$, and with vital dynamics~$\mu>0$.  Without loss of generality, we consider the case~$\sigma_{21}\le \sigma_{12}$.
Thus, the parameter space is 
\begin{equation}\label{eq:parameterspace}
0<\sigma_{21}\le \sigma_{12}\le 1,\qquad \mu>0,\quad \beta_i>0,\quad \gamma_i>0, \qquad i=1,2.
\end{equation}

\subsection{Basic reproduction number}\label{sec:R0}
The basic reproduction number, corresponding to model~\eqref{eq:model}, equals~\cite{chung2016dynamics,castillo1989epidemiological}
\[
\mathcal{R}_0=\max\{\mathcal{R}_1,\mathcal{R}_2\},
\]
where for~$i=1,2,$
\begin{equation}\label{eq:def_Ri}
\mathcal{R}_i:=\frac{\beta_i}{\gamma_i+\mu},
\end{equation}
The disease free equilibrium,~$S=1$ and~$I_i=J_i=R_i=0$ for~$i=1,2$, is stable if and only if~$\mathcal{R}_0<1$~\cite{chung2016dynamics,castillo1989epidemiological}.

\subsection{Single strain (boundary) steady-states}
We consider steady-states of~\eqref{eq:model} with only one strain,~$I_i>0$ and~$I_j=0$,~$i\ne j$.  

\begin{lemma}[Single strain steady-states~$\phi^{EE}_i$]\label{lem:singleStrainSteadyState}
Let~$\mu,\beta_1,\beta_2,\gamma_1,\gamma_2,\sigma_{12}$ and~$\sigma_{21}$ be parameters that satisfy~\eqref{eq:parameterspace}, and let us define~$\mathcal{R}_1$,~$\mathcal{R}_2$, by~\eqref{eq:def_Ri}.

If~$\mathcal{R}_i>1$ for~$i=1$ or ~$i=2$, then there exists a single steady-state solution~$\phi^{EE}_i$ of~$\eqref{eq:model}$ with~$I_i>0$ and~$I_j=0$,~$j\ne i$.
This solution satisfies
\begin{enumerate}
\item The compartment sizes of~$\phi^{EE}_i$ are
\[
S=\frac1{\mathcal{R}_i},\quad I_i=\frac{\mu}{\mu+\gamma_i}\frac{\mathcal{R}_i-1}{\mathcal{R}_i},\quad R_i=\frac{\gamma_i}{\mu+\gamma_i}\frac{\mathcal{R}_i-1}{\mathcal{R}_i},\quad I_j=J_i=J_j=R_j=R_3=0.
\]
\item $\phi^{EE}_i$ is linearly stable if and only if
\[
\mathcal{R}_j<\frac{(\gamma_i+\mu)\mathcal{R}_i}{(1+\sigma_{ij}(\mathcal{R}_i-1))\gamma_i+\mu},\quad j\ne i.
\]
\end{enumerate}
Otherwise, if~$\mathcal{R}_i\le 1$, the system~$\eqref{eq:model}$ does not have a steady-state solution with~$I_i>0$ and~$I_j=0$,~$j\ne i$.
\end{lemma}
\begin{proof}
    See~\cite[Section 3]{castillo1989epidemiological} or~\cite[Section 2]{chung2016dynamics}.
\end{proof}

\section{Coexistence (interior) steady-state}\label{sec:phiCE}
In what follows, we aim to find when the system~\eqref{eq:model} gives rise to a steady-state of coexistence~$\phi_{CE}$ for which~$I_1>0$ and~$I_2>0$,   
and to find the compartment sizes at this equilibrium.  
The steady-state~$\phi_{CE}$ consists of eight variables~$(S^*,I_1^*,I^*_2,J^*_1,J^*_2,R^*_1,R^*_2,R^*_3)$ that are a solution of a nonlinear algebraic system of eight equations.  One can rarely expect to find an explicit solution for such a system, unless in carefully selected cases. Nevertheless, we show how the algebraic system of eight equations can be reduced to a cubic equation and utilize this reduction to prove that a steady-state of coexistence uniquely exists in a region of the parameter space.
\begin{proposition}[Reduction to a cubic  equation for~$\phi_{CE}$]\label{prop:reduceSystem}
    Let~$\mu$,~$\{\beta_i,\gamma_i,\}_{i=1}^2$,~$\sigma_{12}$ and~$\sigma_{21}$ be parameters that satisfy~\eqref{eq:parameterspace}, and let
\[
\phi_{CE}=(S^*,I_1^*,I^*_2,J^*_1,J^*_2,R^*_1,R^*_2,R_3^*)
\]
be a steady-state solution of~\eqref{eq:model} with~$I_1^*\ne0$ and~$I_2^*\ne0$.  
Then,~$S^*$ is a solution of the cubic equation
\begin{subequations}\label{eq:cubic_eq4S_with_ci}
\begin{equation}\label{eq:cubic_eq4S}
P_S(S^*)=c_3(S^*)^3+c_2 (S^*)^2+c_1 S^*+c_0=0,
\end{equation}
with
\begin{equation}\label{eq:ci}
\begin{split}
c_3&=[\sigma_{21}(1-\sigma_{12})\gamma_1+\sigma_{12}(1-\sigma_{21})\gamma_2-\mu(\sigma_{21}+\sigma_{12}-\sigma_{12}\sigma_{21})]\mu\mathcal{R}_1\mathcal{R}_2,\\
c_2&=-[\sigma_{12}(1-\sigma_{21} )\mathcal{R}_1+\sigma_{21}(1-\sigma_{12} )\mathcal{R}_2]\gamma_1\gamma_2-\\&\left[(\gamma_1 + \gamma_2)\sigma_{12}\sigma_{21}\mathcal{R}_1\mathcal{R}_2 -\gamma_1\sigma_{21}\mathcal{R}_1-\gamma_2\sigma_{12}\mathcal{R}_2+\right.\\&\left.(\gamma_1+\gamma_2)(\sigma_{21}-1)\sigma_{12}\mathcal{R}_1+(\gamma_1+\gamma_2)(\sigma_{12}-1)\sigma_{21}\mathcal{R}_2\right]\mu-\\&
[\sigma_{12}\sigma_{21}\mathcal{R}_1\mathcal{R}_2-(\sigma_{12} + \sigma_{21}-\sigma_{12}\sigma_{21}  )(\mathcal{R}_1+\mathcal{R}_2)]\mu^2,\\
c_1&=(\gamma_1+\mu)(\gamma_2+\mu)[\sigma_{12}\sigma_{21}(\mathcal{R}_1+\mathcal{R}_2) + \sigma_{12}\sigma_{21} - \sigma_{12} -\sigma_{21}]\\
c_0&=-(\gamma_1+\mu)(\gamma_2+\mu)\sigma_{12}\sigma_{21}.
\end{split}
\end{equation}
\end{subequations}
The remaining compartments,~$I_1^*,I^*_2,J^*_1,J^*_2,R^*_1,R_2^*$, and~$R^*_3$, are defined by~$S^*$,
\begin{subequations}\label{eq:reducedSystem}
\begin{equation}\label{eq:Ii}
I_i^*(S^*)=\frac{(\gamma_j+\mu)(1-\mathcal{R}_jS^*)R_j^*(S^*) + \gamma_j\mathcal{R}_j S^* R_i^*(S^*)}{\gamma_1\gamma_2(1 - (\mathcal{R}_1+\mathcal{R}_2)S^*)+(1-\mathcal{R}_1S^*)(1-\mathcal{R}_2S^*)(\gamma_1 + \gamma_2+\mu)\mu}\mathcal{R}_iS^*\mu,\quad i=1,2.
\end{equation}
where
\begin{equation}\label{eq:Ji}
J_1^*(S^*)=\frac{\gamma_2 I_2^*(S^*)-\mu R_2^*(S^*)}{\gamma_1+\mu},\quad J_2^*(S^*)=\frac{\gamma_1 I_1^*(S^*)-\mu R_1^*(S^*)}{\gamma_2+\mu},
\end{equation}
\begin{equation}\label{eq:Ri}
R^*_1(S^*)=\frac{1-\mathcal{R}_2S^*}{\mathcal{R}_2\sigma_{12}},\quad R_2^*(S^*)=\frac{1-\mathcal{R}_1S^*}{\mathcal{R}_1\sigma_{21}},
\end{equation}
and~$\mu R_3^*=\gamma_1 J_1^*(S^*)+\gamma_2 J_2^*(S^*)$.
\end{subequations}
\end{proposition}
\begin{proof}
The steady-state solution satisfies~$J_1^\prime(t)+R_2^\prime(t)=0=J_2^\prime(t)+R_1^\prime(t)$.  Thus, Equations~(\ref{eq:eqJ1}- \ref{eq:eqR2}) imply~\eqref{eq:Ji}.
 Since~$\mu>0$, Equation~\eqref{eq:eqS} implies that~$S^*\ne0$.  Thus,~$I_i^*\ne0$ implies that~$I_i^*+J_i^*\ne0$,~$i=1,2$, see~(\ref{eq:eqI1},\ref{eq:eqI2}).  Hence, 
Equations~(\ref{eq:eqI1},\ref{eq:eqI2}) imply
\begin{equation}\label{eq:IiJi_rel1}
\frac{I_1^*}{I^*_1+J^*_1}=\mathcal{R}_1S^*,\quad \frac{I^*_2}{I^*_2+J^*_2}=\mathcal{R}_2S^*.
\end{equation}
We next derive relation~\eqref{eq:Ri}: 
By~\eqref{eq:IiJi_rel1},
\begin{equation}\label{eq:IiJi_rel2}
\frac{J_1^*}{I_1^*+J_1^*}=1-\mathcal{R}_1S^*,\quad \frac{J^*_2}{I^*_2+J^*_2}=1-\mathcal{R}_2S^*.
\end{equation}
Substituting the above expression in~(\ref{eq:eqJ1},\ref{eq:eqJ2}) defines~$R^*_i$ as a linear function of~$S^*$ for~$\sigma_{ij}>0$ 
\[
\mathcal{R}_1\sigma_{21}R_2^*=\frac{J_1^*}{I_1^*+J^*_1}=1-\mathcal{R}_1S^*,\quad \mathcal{R}_2\sigma_{12}R^*_1=\frac{J_2^*}{I_2^*+J_2^*}=1-\mathcal{R}_2S^*.
\]
Solving the system~(\ref{eq:Ji},\ref{eq:Ri},\ref{eq:IiJi_rel1})  yields~\eqref{eq:Ii} for~$I_1^*(S^*)$ and~$I_2^*(S^*)$.
Substituting~\eqref{eq:reducedSystem} in~\eqref{eq:eqS} yields the cubic equation~\eqref{eq:cubic_eq4S_with_ci} for~$S^*$.
\end{proof}

Proposition~\ref{prop:reduceSystem} explicitly defines the coexistence steady-state solution~$\phi_{CE}$, if exists, in terms of a solution to a cubic equation.  The following theorem builds upon this result to prove that a steady-state of coexistence uniquely exists in a region of the parameter space.
\begin{theorem}[Existence and uniqueness of~$\phi_{CE}$]\label{theo:existenceCE}
        Let~$\mu$,~$\{\beta_i,\gamma_i\}_{i=1}^2$,~$\sigma_{12}$ and~$\sigma_{21}$ be parameters that satisfy~\eqref{eq:parameterspace}.
    Define
    \begin{subequations}\label{eq:Omega}
    \begin{equation}
        \Omega^\mu=\Omega^\mu_1\cap \Omega^\mu_2,
    \end{equation}    
    where
    \begin{equation}\label{eq:omega_i}
    \Omega^\mu_i=\left\{\mathcal{R}_i>1,\,\mathcal{R}_j>\frac{(\gamma_i+\mu)\mathcal{R}_i}{(1+\sigma_{ij}(\mathcal{R}_i-1))\gamma_i+\mu}\right\},\quad j\ne i,
    \end{equation}
        \end{subequations}
        and~$\mathcal{R}_i$ is defined by~\eqref{eq:def_Ri}.
    If~$(\mathcal{R}_1,\mathcal{R}_2)\in \Omega^\mu$, then the system~\eqref{eq:model} has a unique steady-state solution~$\phi_{CE}$ satisfying
        \begin{enumerate}
    \item Feasibility: Satisfies~\eqref{eq:boundOnVariables}.
    \item Coexistence:~$I_1^*>0$ and~$I_2^*>0$.
    \end{enumerate}
\end{theorem}
\begin{proof}
    By~\eqref{eq:Ri},
    \[
    R_i^*>0\quad \iff\quad S^*<\frac{1}{\mathcal{R}_j},\qquad i=1,2,\quad j\ne i.
    \]
    We now show that there if~$(\mathcal{R}_1,\mathcal{R}_2)\in \Omega^\mu$, then there exists a unique solution of~\eqref{eq:cubic_eq4S_with_ci} satisfying
\begin{equation}\label{eq:boundsS*}
0<S^*<\min_{i=1,2}\frac{1}{\mathcal{R}_i}.
\end{equation}
To do so, we show that the cubic polynomial~$P_S(s)$ for~$S^*$,~\eqref{eq:cubic_eq4S_with_ci}, changes sign between~$s=0$ and~$s=\min_{i=1,2}\frac{1}{\mathcal{R}_i}$.
Indeed,~$P_S(0)<0$ since~$c_0<0$, see~\eqref{eq:ci}.
Additionally, 
\begin{equation}\label{eq:Ps_at_1Ri}
P_S\left(\frac{1}{\mathcal{R}_i}\right)=\frac{(1+\sigma_{ij}(\mathcal{R}_i-1))\gamma_i+\mu}{\mathcal{R}_i^2}\frac{\gamma_j\sigma_{ji}(\gamma_1+\mu)(\gamma_2+\mu)}{\mathcal{R}_i^2}\left[\mathcal{R}_j-\frac{(\gamma_i+\mu)\mathcal{R}_i}{(1+\sigma_{ij}(\mathcal{R}_i-1))\gamma_i+\mu}\right].
\end{equation}
Therefore, if~$(\mathcal{R}_1,\mathcal{R}_2)\in \Omega^\mu$ then~$P_S\left(\frac{1}{\mathcal{R}_i}\right)>0$ for~$i=1,2$.
This implies that~$P_S$ has a  root,~$S^*$, that satisfies~\eqref{eq:boundsS*}.  We now show that this is the only root of~$P_S$ that satisfies~\eqref{eq:boundsS*}. 
Indeed, since~$c_3<0$ and~$c_0<0$, see~\eqref{eq:ci},
\[
\lim_{s\to\pm \infty} P_S(s)=\mp \infty.
\]
Therefore,~$P_S(s)$ changes sign in the interval~$s<0$ and in the interval~$s>\max_{i=1,2}\mathcal{R}_i^{-1}$ so that it has two additional roots~$S_\pm$ that satisfy
\[
S_-<0,\quad S_+>\max_{i=1,2}\frac{1}{\mathcal{R}_i}.
\]
Therefore, out of the three roots of~$P_S(s)$, the root~$S^*$ is the only root that satisfies~\eqref{eq:boundsS*}.

Next, we now that if~$(\mathcal{R}_1,\mathcal{R}_2)\in \Omega^\mu_j$ then~$I_i>0$ for~$i=1,2$ and~$j\ne i$.  
To do so, we first present a parametrization of~$\Omega^\mu_2$. Let us define
\begin{subequations}\label{eq:dr_ds_rel}
\begin{equation}\label{eq:def_dr1}
d_{r_1}:=\frac{(1+\sigma_{21}(\mathcal{R}_2-1))\gamma_1+\mu}{(\gamma_2+\mu)\mathcal{R}_2}\left[\mathcal{R}_1-\frac{(\gamma_2+\mu)\mathcal{R}_2}{(1+\sigma_{21}(\mathcal{R}_2-1))\gamma_1+\mu}\right],
\end{equation}
so that for any~$\mathcal{R}_2>1$,
\[
(\mathcal{R}_1(d_{r_1},\mathcal{R}_2),\mathcal{R}_2)\in \Omega^\mu_2\iff d_{r_1}>0,\qquad \mathcal{R}_1(d_{r_1},\mathcal{R}_2)=\frac{(\gamma_2+\mu)\mathcal{R}_2}{(1+\sigma_{21}(\mathcal{R}_2-1))\gamma_1+\mu}(1+d_{r_1}).
\]
Substituting~$\mathcal{R}_1=\mathcal{R}_1(d_{r_1},\mathcal{R}_2)$ in~\eqref{eq:cubic_eq4S_with_ci} yields the following relation between~$d_{r_1}$,~$\mathcal{R}_2$ and~$S^*$,
\begin{equation}\label{eq:dr_ds_rel}
d_{r_1}=-\frac{(1+\sigma_{21}(\mathcal{R}_2-1))\gamma_1+\mu}{\gamma_2+\mu}\frac{1-\mathcal{R}_2S^*}{\frac{\partial c_3}{\partial \mathcal{R}_1}(S^*)^2 +\frac{\partial c_2}{\partial \mathcal{R}_1}S^*  + \frac{\partial c_1}{\partial \mathcal{R}_1}}P_S\left(S^*;\mathcal{R}_1=\frac{(\gamma_2+\mu)\mathcal{R}_2}{(1+\sigma_{21}(\mathcal{R}_2-1))\gamma_1+\mu},\mathcal{R}_2\right).
\end{equation}
Note that the coefficients~$c_i$ of~\eqref{eq:cubic_eq4S} are linear in~$\mathcal{R}_1$, and therefore their partial derivative with respect to~$\mathcal{R}_1$ are independent of~$\mathcal{R}_1$.  Thus, the relation~\eqref{eq:dr_ds_rel} explicit defines~$d_{r_1}$ in terms of~$\mathcal{R}_2$ and~$S^*$.  
\end{subequations}
Substituting~$\mathcal{R}_1=\mathcal{R}_1(d_{r_1},\mathcal{R}_2)$ into~\eqref{eq:Ii} yields
\begin{equation}\label{eq:simplifiedI1}
I_1^*=\frac{\mu}{\sigma_{12}}\frac{S^*+\sigma_{12}(1-S^*)}{\gamma_1+\mu(1-\mathcal{R}_2 S^*)}(1-\mathcal{R}_2 S^*).
\end{equation}
Therefore,~\eqref{eq:boundsS*} implies~$I_1^*>0$.  Similarly, if~$(\mathcal{R}_1,\mathcal{R}_2)\in \Omega^\mu_1$ then~$I_2^*>0$.

Next, we show that if~$(\mathcal{R}_1,\mathcal{R}_2)\in \Omega^\mu$ then~$J_i^*>0$.
Indeed, substituting~\eqref{eq:simplifiedI1} into~\eqref{eq:Ji} yields
\begin{equation}\label{eq:simplifiedJ1}
J_i^*=\mu\frac{A_i}{(\gamma_i+\mu)\sigma_{ji}\mathcal{R}_i}\frac{1-\mathcal{R}_i S^*}{\gamma_j+\mu(1-\mathcal{R}_i S^*)},\quad j\ne i,
\end{equation}
where
\[
A_i:=\mathcal{R}_i\gamma_j\sigma_{ji}(1 - S^*) - (\gamma_j + \mu)(1-\mathcal{R}_iS^*),\quad j\ne i.
\]
For~$(\mathcal{R}_1,\mathcal{R}_2)\in\Omega^\mu$, the bound~\eqref{eq:boundsS*} implies that $J_i^*>0$ if and only if~$A_i>0$ where~$i=1,2$.
By~\eqref{eq:Ri},
\begin{equation}\label{eq:Ai_positive}
A_i=\mathcal{R}_i\sigma_{ji}[\gamma_j(1-S^*)- (\gamma_j + \mu) R_j^*]>\gamma_j\mathcal{R}_i\sigma_{ji}(1-S^*-R_j^*)>0,
\end{equation}
where the last inequality is due to
\begin{equation}\label{eq:boundp_partial_sum}
\begin{split}
1-(S^*+R_1^*+R_2^*)&=\left(1-\frac1{\mathcal{R}_2\sigma_{12}}-\frac1{\mathcal{R}_2\sigma_{21}}\right)-\left(1-\frac1{\sigma_{12}}-\frac1{\sigma_{21}}\right)S^*>\\&\left(1-\frac1{\sigma_{12}}-\frac1{\sigma_{21}}\right)(1-S^*)>0,
\end{split}    
\end{equation}
where~\eqref{eq:Ri} applies in the first equality.
Moreover, when~$J_i>0$ for~$i=1,2,$ then~$R_3^*=\gamma_1 J_1^*+\gamma_2 J_2^*>0$.

Finally, summation of all equations of~\eqref{eq:model} implies that
\[
S^*+I_1^*+I_2^*+J_1^*+J_2^*+R_1^*+R_2^*+R_3^*=1.
\]
Thus, for~$(\mathcal{R}_1,\mathcal{R}_2)\in \Omega^\mu$, the steady-state solution~$\phi_{CE}$ satisfies~\eqref{eq:boundOnVariables}.
\end{proof}

Theorem~\ref{theo:existenceCE} proves the existence and uniqueness of a steady state with co-existing strains in the parameter regime~$\Omega^\mu$ defined by \eqref{eq:Omega}.  Note that~$\Omega^\mu$ is the regime in which the other steady-states are unstable, see Section~\ref{sec:R0} and Lemma~\ref{lem:singleStrainSteadyState}.  In the subsequent section, we will prove, for sufficiently small~$\mu>0$, the non-existence of a steady-state with coexisting strains outside the parameter regime~$\Omega^\mu$.  Hence, for small enough~$\mu>0$, the system~\eqref{eq:model} is not bi-stable.  We conjecture that this result is valid for all~$\mu>0$.  

The proof of Theorem~\ref{theo:existenceCE} encompasses several findings that will aid the following analysis. The next Lemma serves as a complement to Proposition~\ref{prop:reduceSystem}.
\begin{lemma}\label{lem:simplifiedIiJi}
Under the conditions of Proposition~\ref{prop:reduceSystem},~$I_1^*$ and~$I_2^*$ satisfy
\begin{equation}\label{eq:simplifiedI1_lemma}
I_1^*=\frac{\mu}{\sigma_{12}}\frac{S^*+\sigma_{12}(1-S^*)}{\gamma_1+\mu(1-\mathcal{R}_2 S^*)}(1-\mathcal{R}_2 S^*),\quad I_2^*=
\frac{\mu}{\sigma_{21}}\frac{S^*+\sigma_{21}(1-S^*)}{\gamma_2+\mu(1-\mathcal{R}_2 S^*)}(1-\mathcal{R}_1 S^*),
\end{equation}
and~$J_1^*$ and~$J_2^*$ satisfy
\[
J_i^*=\mu\frac{A_i}{(\gamma_i+\mu)\sigma_{ji}\mathcal{R}_i}\frac{1-\mathcal{R}_i S^*}{\gamma_j+\mu(1-\mathcal{R}_i S^*)},\quad j\ne i,
\]
where
\[
A_1:=\mathcal{R}_1\gamma_2\sigma_{21}(1 - S^*) - (\gamma_2 + \mu)(1-\mathcal{R}_1S^*),\quad 
A_2:=\mathcal{R}_2\gamma_1\sigma_{12}(1 - S^*) - (\gamma_1 + \mu)(1-\mathcal{R}_2S^*).
\]
\end{lemma}
\begin{proof}
    This follows directly from the proof of Theorem~\ref{theo:existenceCE}, see~\eqref{eq:simplifiedI1}.  Note that the result also applies for~$(\mathcal{R}_1,\mathcal{R}_2)\not\in \Omega^\mu$.
\end{proof}

The following lemma highlights useful inequalities.
\begin{lemma}\label{lem:1minuss_0_lower_bound_gnrl_mu}
Under the conditions of Theorem~\ref{theo:existenceCE}. If~$(\mathcal{R}_1,\mathcal{R}_2)\in \Omega^\mu$, then~$S^*$ satisfies 
\begin{equation}\label{eq:1minuss_0_lower_bound_gnrl_m}
1-\mathcal{R}_iS^*>0,\quad A_i>0, \qquad i=1,2.
\end{equation}
where
\[
A_1:=\mathcal{R}_1\gamma_2\sigma_{21}(1 - S^*) - (\gamma_2 + \mu)(1-\mathcal{R}_1S^*),\quad 
A_2:=\mathcal{R}_2\gamma_1\sigma_{12}(1 - S^*) - (\gamma_1 + \mu)(1-\mathcal{R}_2S^*).
\]
\end{lemma}
\begin{proof}
The proof of Theorem~\ref{theo:existenceCE} directly shows that~$1-\mathcal{R}_is_0>0$ for~$i=1,2$, see~\eqref{eq:boundsS*}.  It also directly shows that~$A_i>0$, see~\eqref{eq:Ai_positive}.
\end{proof}
\section{Existence and uniqueness of coexistence steady-state for~$\mu\ll\gamma_i$}\label{sec:CE_perturbation}
In what follows, we focus on the asymptotic regime
\begin{equation}\label{eq:asymptotic_regime}
\mu\ll \gamma_i,\quad i=1,2,
\end{equation}
that corresponds to a case in which the birth and mortality rate~$\mu$ is significantly smaller than the recovery rates~$\gamma_1$ and~$\gamma_2$.  Such a separation of time scales arises for example, in Influenza, Varicella, COVID-19, and Rubeola~\cite{CDCyellowbook2024} for which the characteristic infectious period is on a scale of a week compared to a human life span of 70-75 years.   

The following proposition represents the results of Section~\ref{sec:phiCE} in the asymptotic regime~\eqref{eq:asymptotic_regime}, and supplements to them.
\begin{proposition}[Steady-state of coexistence]\label{prop:CE_small_mu}
    Let~$\{\beta_i,\gamma_i\}_{i=1}^2$,~$\sigma_{12}$ and~$\sigma_{21}$ be parameters that satisfy~\eqref{eq:parameterspace}, and let~$\mathcal{R}_i$ be defined by~\eqref{eq:def_Ri}.
    Then, there exists a sufficient small~$\mu^*<\min\{\gamma_1,\gamma_2\}$ such that for all~$0<\mu<\mu^*$
\begin{itemize}
    \item {\bf Existence and uniqueness of a feasible steady-state solution of coexistence:}
    If~$(\mathcal{R}_1,\mathcal{R}_2)\in \Omega^{\mu}$, then there exists a unique steady-state solution of~\eqref{eq:model} satisfying \begin{enumerate}
    \item Feasibility: Satisfies~\eqref{eq:boundOnVariables}.
    \item Coexistence:~$I_1^*>0$ and~$I_2^*>0$.
    \item Approximate values: The steady-state solution satisfies
    \begin{subequations}\label{eq:approxS*_andIi*}
\begin{equation}
S^*=s_0+\mathcal{O}(\mu),\quad I_i^*=\mu b_i+\mathcal{O}(\mu^2),\quad i=1,2,
\end{equation}
where~$s_0$ is a solution of the quadratic equation
\begin{equation}\label{eq:quadratic_eq4S}
p_s(s_0)=as_0^2+bs_0+c=0,
\end{equation}
with 
\begin{equation}\label{eq:quadratic_eq4S_coeffs}
a=\sigma_{12}(1-\sigma_{21})\mathcal{R}_1+\sigma_{21}(1-\sigma_{12})\mathcal{R}_2,\quad b=\sigma_{12}\sigma_{21}\left(\mathcal{R}_1+\mathcal{R}_2\right)-s,\quad c=-\sigma_{12}\sigma_{21},
\end{equation}
\[
s:=\sigma_{12}+\sigma_{21}-\sigma_{12}\sigma_{21}=1-(1-\sigma_{12})(1-\sigma_{21}),
\]
and~$b_1$ and~$b_2$ satisfy
\begin{equation}\label{eq:b_i_simplified}
b_1=\frac{s_0+\sigma_{12}(1-s_0)}{\gamma_1\sigma_{12}}(1-\mathcal{R}_2s_0),\qquad
b_2=\frac{s_0+\sigma_{21}(1-s_0)}{\gamma_2\sigma_{21}}(1-\mathcal{R}_1s_0).
\end{equation}
\end{subequations}
    \end{enumerate}
    \item {\bf Non-existence:}
    If~$(\mathcal{R}_1,\mathcal{R}_2)\not\in \Omega^\mu$, then a steady-state solution of~\eqref{eq:model} with~$I_1^*>0$ and~$I_2^*>0$ does not exist.
    \end{itemize}
\end{proposition}
\begin{proof}
By Theorem~\ref{theo:existenceCE}, if~$(\mathcal{R}_1,\mathcal{R}_2)\in \Omega^\mu$ then there exists a unique feasible steady-state solution of~\eqref{eq:model} with strictly positive compartment sizes. 
To derive the leading order approximation~\eqref{eq:approxS*_andIi*} of~$S^*$ and~$I_i^*$, we substitute~$\mu=0$ in~\eqref{eq:cubic_eq4S_with_ci} and in the coefficient of~$\mu$ in~\eqref{eq:simplifiedI1_lemma}.


Next, we prove the nonexistence of a feasible steady-state solution of~\eqref{eq:model} with coexisting strains when~$(\mathcal{R}_1,\mathcal{R}_2)\not\in\Omega^\mu$.  To do so, we show that if~$(\mathcal{R}_1,\mathcal{R}_2)\not\in\Omega^\mu$ then~$s_0>\max_{i=1,2} \mathcal{R}_i$.  Thus, by~\eqref{eq:b_i_simplified}, for sufficiently small~$\mu$,~$I_1^*<0$ or~$I_2^*<0$.
Indeed, if~$(\mathcal{R}_1,\mathcal{R}_2)\not\in\Omega^\mu_i$ then
\[
p_s(1/{\mathcal{R}_i})<0,
\]
see~\eqref{eq:Ps_at_1Ri}.
Additionally~$\lim_{s\to\infty}p_s(s)=\infty$.  Therefore, in this case, $p_s$ has a root satisfying~$s_0>\max_{i=1,2} \mathcal{R}_i$.  Since~$ac<0$, see~\eqref{eq:quadratic_eq4S_coeffs}, this root is the only positive root of~$p_s$.

\end{proof}
Studying the systems' behavior in an asymptotic regime enables us to use approximation methods that open the way to analysis.
For example, while in the general case,~$S^*$ satisfies a cubic equation, see Proposition~\ref{prop:reduceSystem}, in the asymptotic regime~\eqref{eq:asymptotic_regime},~$S*$ can be approximated by a solution of a simpler quadratic equation.  Using such approximations, Proposition~\ref{prop:CE_small_mu} supplements Theorem~\ref{theo:existenceCE} by proving that a steady state of coexistence does not exist if~$(\mathcal{R}_1,\mathcal{R}_2)\not\in \Omega^\mu$.
This implies that for sufficiently small~$\mu$,~$\Omega^\mu$ is the maximal subset of the~$(\mathcal{R}_1,\mathcal{R}_2)$ plane in which a feasible steady state with coexisting strains exists.
Since the other steady-states are stable only outside of~$\Omega^\mu$, this result also implies that, for sufficiently small~$\mu$, the system~\eqref{eq:model} does not give rise to bi-stability.
\begin{lemma}[No bi-stability]\label{lem:no_bistability}
    Let~$\{\beta_i,\gamma_i,\}_{i=1}^2$,~$\sigma_{12}$ and~$\sigma_{21}$ be parameters that satisfy~\eqref{eq:parameterspace}.  
    If~$\phi^{EE}_1$ or~$\phi^{EE}_2$ are linearly stable for all~$\mu>0$, then, for small enough~$\mu<\min\{\gamma_1,\gamma_2\}$, a steady-state~$\phi_{CE}$ of coexistence does not exist.
\end{lemma}

\section{Stability of the coexistence steady-state}\label{sec:stability_CE}
Proposition~\ref{prop:CE_small_mu} proves
the existence and uniqueness of a coexistence steady-state~$\phi_{CE}$ in the parameter regime~$\Omega^\mu$ for sufficiently small~$\mu$.  We now consider the stability of~$\phi_{CE}$.  
One can readily determine the stability of~$\phi_{CE}$ by computing the Jacobian of~\eqref{eq:model} at~$\phi_{CE}$, or at a sufficiently accurate approximation of it in the regime~\eqref{eq:asymptotic_regime}, and studying its eigenvalues. However, the resulting Jacobian is a $7\times7$ matrix.  Since the characteristic polynomial is a septic equation, it cannot be solved analytically except in special cases.  Bypassing this problem by applying the Routh-Horowitz criteria to the corresponding characteristic polynomial of degree seven is a notorious, almost impractical, task.  Instead, in what follows, we continue to exploit the separation of time scales,~\eqref{eq:asymptotic_regime}, and derive an approximation of the roots of the characteristic polynomial at the minimal accuracy required to determine the stability of~$\phi_{CE}$.

The characteristic polynomial of the Jacobian of~\eqref{eq:model} at~$\phi_{CE}$ takes the form
\begin{equation}\label{eq:Plambda}
P(\lambda)=\lambda^5(\lambda+\gamma_1)(\lambda+\gamma_2)+\mu\lambda^3 \sum_{j=0}^4\lambda^jP_{1j}+\mu^2 \lambda \sum_{j=0}^6\lambda^jP_{2j}+\mu^3\sum_{j=0}^7\lambda^jP_{3j}+\mathcal{O}(\mu^4),\qquad P_{i0}\ne0,\quad i=1,2,3,
\end{equation}
where the coefficients,~$P_{ij}=P_{ij}(\mathcal{R}_1,\mathcal{R}_2;\gamma_1,\gamma_2,\sigma_{12},\sigma_{21})$,  depend on the problem parameters. The explicit expressions for the coefficients,~$P_{ij}$, are available via symbolic computations and are generally long and cumbersome.  For conciseness, in what follows, we will provide explicit expressions only by context and for the coefficients whose explicit form is used in the stability analysis.  

The polynomial~$P(\lambda)$ has seven roots~$\{\lambda_j\}_{j=1}^7$.  Two roots satisfy
\begin{subequations}\label{eq:lambda_123}
    \begin{equation}
        \lambda_j=-\gamma_j+\mathcal{O}(\mu),\quad j=1,2,
    \end{equation}
and, thus, are negative to leading order.
An additional root satisfies
    \begin{equation}
        \lambda_3=-\frac{P_{30}}{P_{20}}\mu+\mathcal{O}(\mu^2)=-\frac{2as_0+b}{as_0-c}\mu+\mathcal{O}(\mu^2),
    \end{equation}
where~$P_{20}$ and~$P_{30}$ are coefficients in~$P(\lambda)$, see~\eqref{eq:Plambda}, and their ratio is expressed in terms of the coefficients,~$a,b,c$, of~\eqref{eq:quadratic_eq4S} which are given by~\eqref{eq:quadratic_eq4S_coeffs}.
To the leading order, this root is negative, since~$c<0$ and since $s_0>-b/2a$.
\end{subequations}
Dominant balance (see, e.g.,~\cite{holmes2012introduction}) shows that the four additional roots of~$P(\lambda)$ are of~$\mathcal{O}(\sqrt\mu)$, 
\begin{equation}\label{eq:lambda_47}
    \lambda_j=\sqrt{\mu} \,r_j+o(\sqrt\mu),\quad j=4,5,\cdots,7.
\end{equation}
Substituting expansion~\eqref{eq:lambda_47} into~$P(\lambda)$ yields
\begin{equation}\label{eq:Plambda_residual}
P(\sqrt{\mu} \,r_j)=(\gamma_1\gamma_2 r_j^4+P_{10}r_j^2+P_{20})r_j\mu\sqrt\mu+\mathcal{O}(\mu^2).
\end{equation}
Thus,~$r_j$ is a solution of
\begin{equation}\label{eq:eq4ri}
\gamma_1\gamma_2 r_j^4+P_{10}r_j^2+P_{20}=0,
\end{equation}
where the coefficients~$P_{10}$ and~$P_{20}$ of~$P(\lambda)$, see~\eqref{eq:Plambda}, are given by
\begin{subequations}\label{eq:P10_P20}
\begin{equation}
P_{10}=B_1+B_2,\quad P_{20}=\frac{B_1B_2}{\gamma_1\gamma_2}\left[1-\frac{(1-\sigma_{12})(1-\sigma_{21})\mathcal{R}_1\mathcal{R}_2s_0^2}{(\sigma_{12}+(1-\sigma_{12})s_0\mathcal{R}_2)(\sigma_{21}+(1-\sigma_{21})s_0\mathcal{R}_1)}\right],
\end{equation}
where
\begin{equation}\label{eq:Bi}\begin{split}
B_1:=&\frac{\gamma_1\gamma_2}{\sigma_{21}}[\mathcal{R}_1s_0+(1-\mathcal{R}_1s_0)\sigma_{21}]A_1^0,\qquad
B_2:=\frac{\gamma_2\gamma_1}{\sigma_{12}}[\mathcal{R}_2s_0+(1-\mathcal{R}_2s_0)\sigma_{12}]A_2^0,
\end{split}
\end{equation}
and where~$A_i^0=A_i(\mu=0)$, see Lemma~\ref{lem:1minuss_0_lower_bound_gnrl_mu},
\[
A_1^0:=\gamma_1[\sigma_{21}\mathcal{R}_1(1-s_0)-(1-\mathcal{R}_1s_0)],\quad 
A_2^0:=\gamma_2[\sigma_{12}\mathcal{R}_2(1-s_0)-(1-\mathcal{R}_2s_0)].
\]
\end{subequations}
Lemma~\ref{lem:1minuss_0_lower_bound_gnrl_mu} implies that~$B_j>0$,~$j=1,2$, see~\eqref{eq:1minuss_0_lower_bound_gnrl_m}.  Hence,~$P_{10}>0$ and~$P_{20}>0$.  Additionally, the discriminant,~$\Delta$, of~\eqref{eq:eq4ri} satisfies 
\begin{equation}\label{eq:Delta}
\Delta:=P_{10}^2-4\gamma_1\gamma_2 P_{20}=(B_1-B_2)^2+
\frac{4(1-\sigma_{12})(1-\sigma_{21})\mathcal{R}_1\mathcal{R}_2B_1B_2s_0^2}{[\sigma_{12}+(1-\sigma_{12})s_0\mathcal{R}_2][\sigma_{21}+(1-\sigma_{21})s_0\mathcal{R}_1]}\ge0,
\end{equation}
where the inequality is due to~\eqref{eq:parameterspace}, and since~$B_j>0$ for~$j=1,2$. 
Therefore,~$r_j^2=r_\pm^2$ where
\[
r_\pm^2=\frac{-P_{10}\pm \sqrt{\Delta}}{2\gamma_1\gamma_2}<0.
\]
This implies that~$\{r_j\}_{j=4}^7$ are purely imaginary numbers. 
Hence, the stability of~$\phi_{CE}$ is determined by the correction terms of~$\{\lambda_j\}_{j=4}^7$.  
Since the relative residual is~$\mathcal{O}(\sqrt\mu)$, 
\[
\frac{\left|P(\sqrt{\mu} \,r_j)-(\gamma_1\gamma_2 r_j^4+P_{10}r_j^2+P_{20})r_j\mu\sqrt\mu\right|}{|P(\sqrt{\mu} \,r_j)|}=\mathcal{O}(\sqrt\mu),\quad j=4,5,\cdots,7,
\]
see~\eqref{eq:Plambda_residual},
it is reasonable to consider an expansion of~$\lambda_j$ in powers of $\sqrt\mu$ 
\begin{equation}\label{eq:expansion_with_qi}
\lambda_j=\sqrt{\mu} \,r_j+\mu q_j+\mathcal{O}(\mu\sqrt\mu),\quad j=4,5,\cdots,7.
\end{equation}
Substituting~\eqref{eq:expansion_with_qi} into~$P(\lambda)$ and using~$r_j^2=r_\pm^2$ yields that~$q_j=q_\pm$ where
\begin{subequations}\label{eq:qi}
\begin{equation}
q_\pm(\mathcal{R}_1,\mathcal{R}_2;\gamma_1,\gamma_2,\sigma_{12},\sigma_{21})=\mp\frac1{\sqrt{\Delta}}\frac{\pm (\gamma_1+\gamma_2)\Delta^{3/2}+\Delta c^q_2 \pm \sqrt\Delta c^q_1+c^q_0}{16\gamma_1^3\gamma_2^3 r_\pm^2},
\end{equation}
where~$\Delta$ is given by~\eqref{eq:Delta}, and
\begin{equation}\label{eq:cq}
\begin{split}
c^q_2&=2\gamma_1\gamma_2P_{11}-3(\gamma_1+\gamma_2)P_{10},\\
c^q_1&=4\gamma_1^2\gamma_2^2P_{21}-4\gamma_1\gamma_2P_{10}P_{11}+3(\gamma_1+\gamma_2)P_{10}^2,\\
c^q_0&=8\gamma_1^3\gamma_2^3P_{30} - 4\gamma_1^2\gamma_2^2P_{10}P_{21} + 2\gamma_1\gamma_2 P_{10}^2P_{11} -(\gamma_1+\gamma_2)P_{10}^3.
\end{split}
\end{equation}
\end{subequations}
We first note that~$q_j$ may not be defined when~$\Delta=0$.  This is a key point that, subsequently, we will consider systematically.  In the case~$\Delta\ne 0$, then by~\eqref{eq:Delta},~$\Delta>0$, implying that $q_j$ is defined and is it real-valued.
Yet, despite the explicit expression~\eqref{eq:qi} for~$q_\pm$, determining the sign of~$q_\pm$ in the feasible region~$(\mathcal{R}_1,\mathcal{R}_2)\in\Omega^\mu$ of the parameter space requires a notorious study of the cumbersome expression. This effort can be compared to the effort involved in applying the Routh-Horowitz criteria to the septic polynomial~$P(\lambda)$, see~\eqref{eq:Plambda}.  In certain cases,~$P(\lambda)$ reduces to a form that can be analyzed, e.g., due to symmetries.  Such cases are considered in~\cite{chung2016dynamics}, and they all point to linear stability of~$\phi_{CE}$.   Our numerical results, presented subsequently, suggest that~$q_j<0$ in a wide range of scenarios, leading to the conjecture that for a sufficiently small~$\mu<\min_{i=1,2}\gamma_i$, the steady-state~$\phi_{CE}$ is linearly stable when~$\Delta$ does not vanish.  

We now consider the case~$\Delta=0$.  In this case, expansion~\eqref{eq:expansion_with_qi} is invalid.
Let us consider the generalized expansion
\begin{equation}\label{eq:expansion_with_si}
\lambda_j=\sqrt{\mu} \,r_j +\mu^{3/4} s_j+\mu q_j+\mathcal{O}(\mu\sqrt[4]\mu),\quad j=4,5,\cdots,7.
\end{equation}
Substituting~\eqref{eq:expansion_with_si} in~$P(\lambda)$ and equating by powers of~$\mu$, yields at leading order, Equation~\eqref{eq:eq4ri} for~$r_j$, and at the next order
\begin{equation}\label{eq:sqrtDeltasj}
\sqrt{\Delta}r_j^2s_j=0.
\end{equation}
Therefore, as expected, when~$\Delta>0$,~$s_j=0$, and expansion~\eqref{eq:expansion_with_si} reduces to~\eqref{eq:expansion_with_qi}.  
When~$\Delta=0$, the next order yields
    \begin{equation}\label{eq:si}
    s_j^2=-\frac{c^q_0}{16\gamma_1^3\gamma_2^3 P_{10}}\frac1{r_j},
    \end{equation}
    where~$c^q_0$ is given by~\eqref{eq:cq} and~$P_{10}$ is given by~\eqref{eq:P10_P20}.
Note that~$s_j^2$ is proportional to~$1/r_j$, which is purely imaginary, while~$c^q_0$ and~$P_{10}$ are real functions.  Therefore,~$s_j^2$ is purely imaginary.
Consequently, 
\[
s_j=\sqrt{\frac{|c^q_0|}{32\gamma_1^3\gamma_2^3 P_{10} |r_j|}}(\pm 1+i).
\]
Thus, as long as~$c^q_0\ne0$, two of the eigenvalues~$\{\lambda_j\}_{j=4}^7$ have a positive real part.

The above analysis reveals a mechanism of instability that occurs when i)~$\Delta=0$ and ii)~$c^q_0\ne0$.  We now ask under what conditions (i) and (ii) occur simultaneously.
By~\eqref{eq:Bi} and~\eqref{eq:Delta},~$\Delta=0$  if and only if~$B_1=B_2$ and
\[
\frac{4(1-\sigma_{12})(1-\sigma_{21})\mathcal{R}_1\mathcal{R}_2B_1B_2s_0^2}{[\sigma_{12}+(1-\sigma_{12})s_0\mathcal{R}_2][\sigma_{21}+(1-\sigma_{21})s_0\mathcal{R}_1]}=0.
\]
Since~$B_i>0$, and under the assumption~$\sigma_{21}\le \sigma_{12}$, see~\eqref{eq:parameterspace}, it follows that the latter condition requires~$\sigma_{12}=1$.
The following Lemma determines when~$B_1=B_2$ holds in this case,
\begin{lemma}[Characterization of~$\Gamma^*$]\label{lem:B1equalsB2}
If~$\sigma_{12}=1$,
then~$B_1=B_2$ if and only if~$(\mathcal{R}_1,\mathcal{R}_2)\in \Gamma^*$
where~$\Gamma^*$ is the set of points~$(\mathcal{R}_1,\mathcal{R}_2)$ satisfying 
\begin{equation}\label{eq:R2star}
\mathcal{R}_1>0,\quad \mathcal{R}_2=1+\frac{\gamma_1}{\gamma_2\sigma_{21}}\left(\mathcal{R}_1s_0+(1-\mathcal{R}_1s_0)\sigma_{21}\right)\left(\sigma_{21}\mathcal{R}_1(1-s_0)-(1-\mathcal{R}_1s_0)\right).
\end{equation}
Moreover, the curve~$\Gamma^*\subset\Omega^0$
where~$\Omega^0=\Omega^{\mu=0}$, see~\eqref{eq:Omega}.
\end{lemma}
\begin{proof}
    The solution of the algebraic equation~$B_1=B_2$ where~$B_j$ is given by~\eqref{eq:Bi} and~$\sigma_{12}=1$ yields~\eqref{eq:R2star}.  
Let~$(\mathcal{R}_1^*,\mathcal{R}_2^*)\in\Gamma^*$.  Then,
    \[
    \frac{\mathcal{R}_2^*}{1+\sigma_{21}(\mathcal{R}_2^*-1)}=\frac{1+\frac{\gamma_1}{\gamma_2\sigma_{21}}\left(\mathcal{R}_1^*s_0+(1-\mathcal{R}_1^*s_0)\sigma_{21}\right)\left(\sigma_{21}\mathcal{R}_1^*(1-s_0)-(1-\mathcal{R}_1^*s_0)\right)}{1+\frac{\gamma_1}{\gamma_2}\left(\mathcal{R}_1^*s_0+(1-\mathcal{R}_1^*s_0)\sigma_{21}\right)\left(\sigma_{21}\mathcal{R}_1^*(1-s_0)-(1-\mathcal{R}_1^*s_0)\right)}\le 1.
    \]
    Since~$\mathcal{R}_1^*>1$, we have
    \[
    \mathcal{R}_1^*>1\ge\frac{\mathcal{R}_2^*}{1+\sigma_{21}(\mathcal{R}_2^*-1)},
    \]
    so that~$(\mathcal{R}_1^*,\mathcal{R}_2^*)\in \Omega_2^0$.
    Moreover, for~$\sigma_{12}=1$,~$\Omega_1^0\subset \Omega_2^0$.  Thus,~$(\mathcal{R}^*_1,\mathcal{R}^*_2)\in \Omega^0$. 
\end{proof}
Note that Lemma~\ref{lem:B1equalsB2} refers to the parameters regime~$\Omega^0=\Omega^{\mu=0}$.  This is since the curve~$\Gamma^*$ is characterized only in leading order.  

Next, we consider the second necessary condition for instability,~$c_0^q\ne0$.
Substituting~\eqref{eq:R2star} in~\eqref{eq:cq} yields that for every point~$(\mathcal{R}_1^*,\mathcal{R}_2^*)\in\Gamma^*$
\begin{equation}\label{eq:cq0_along_curve}
\begin{split}
c^q_0(\mathcal{R}_1^*,\mathcal{R}_2^*)&=-(1-\sigma_{21})\times\\&\frac{8\gamma_1^6\gamma_2^3(\sigma_{21}+(1-\sigma_{21})\mathcal{R}_1^*s_0)^2[\gamma_2\sigma_{21}+(\mathcal{R}_1^*s_0+\sigma_{21}(1-\mathcal{R}_1^*s_0))A_1^0(\mathcal{R}_1^*)]\mathcal{R}_1^*s_0A_1^0(\mathcal{R}_1^*)}{\sigma_{21}^4}.
\end{split}
\end{equation}
By Lemma~\ref{lem:1minuss_0_lower_bound_gnrl_mu},~$A_1^0>0$, hence~$c^q_0<0$ for~$\sigma_{21}<1$.  
In the case~$\sigma_{21}=1$,~$c^q_0=0$ and the steady-state is linearly stable.  Indeed, in this case,~$s_i=0$ and~$\lambda_i$ behaves as~\eqref{eq:expansion_with_qi} where the numerator and denominator of~$q_j$~\eqref{eq:qi} vanish, so that~$q_j$ is defined at the limit~$\Delta\to0^+$ and equals~$q_j=q_\pm$ where 
\[
q_+=-\frac{\mathcal{R}_1}{2},\quad q_-=-\frac{\mathcal{R}_2}{2}.
\]
We conclude that for a given~$\mathcal{R}_1>1$ if~$\sigma_{12}=1$  then for any value of~$0<\sigma_{21}<1$, for any~$(\mathcal{R}_1,\mathcal{R}_2)\in\Gamma^*$ and for sufficiently small~$\mu>0$, the steady-state~$\phi_{CE}$ is linearly unstable. 

The following Theorem summarizes the results attained in this section and extends them to account for values of~$\sigma_{12}\approx 1$.
\begin{theorem}[Instability along~$\Gamma^*$]\label{theo:main_result}
Let~$\beta_1>\gamma_1>0$,~$\gamma_2>0$,~$0<\sigma_{21}<1$, and~$s_{12}>0$.  
Then, for a sufficiently small~$\mu>0$, and for~$\beta_2$ chosen such that~\eqref{eq:R2star} holds where~$\sigma_{12}=1-\mu s_{12}$, there exists a 
unique steady-state~$\phi_{CE}$ with~$I_1>0$ and~$I_2>0$. Moreover, this steady-state is linearly unstable.
\end{theorem}
\begin{proof}
See Appendix~\ref{app:proof_main_result} for an extension of the results.
\end{proof}
Theorem~\ref{theo:main_result} points to a region in the parameter space~$(\mathcal{R}_1,\mathcal{R}_2,\sigma_{12},\sigma_{21})$ in which the system~\eqref{eq:model} gives rise to a linearly unstable steady state of coexistence.  In this region, the other steady-states of the system are also unstable, see lemma~\ref{lem:singleStrainSteadyState}.  Therefore, one may expect the system to oscillate in this region.  In the subsequent section, we will demonstrate this numerically.

Theorem~\ref{theo:main_result} contradicts~\cite[Theorem 1.2]{chung2016dynamics}.  The reason for this contradiction is that the proof of Theorem 1.2 in~\cite{chung2016dynamics} assumes that the expansion~\eqref{eq:expansion_with_qi} is uniformly valid, while Lemma~\ref{lem:B1equalsB2} and the proof of Theorem~\ref{theo:main_result} show that the expansion~\eqref{eq:expansion_with_qi} becomes invalid in feasible regions of the parameter space, specifically along~$\Gamma^*$, and directly associates the instabilities of~$\phi_{CE}$ with the cases in which the expansion~\eqref{eq:expansion_with_qi} is invalid.

A numerical verification of the approximations of the eigenvalues~$\{\lambda_j\}_{j=1}^7$ presented above is presented in Appendix~\ref{app:verificationLambda}.
\subsection{Numerical results}\label{sec:numerics}
Theorem~\ref{theo:main_result} characterizes a curve~$\Gamma^*$ in the parameter space~$(\mathcal{R}_1,\mathcal{R}_2)$ in which the steady state of coexistence is linearly unstable.  One may expect an oscillatory solution to emerge in this case, yet our analytical results do not provide information on the global dynamics of the system, for example, whether it converges to a periodic solution or displays more complicated dynamics, e.g., chaotic behavior.  Here, we complement the analysis with a numerical study to obtain a better picture of the system dynamics.  Details of the numerical methods are provided in Appendix~\ref{app:numerical_methods}.

\begin{figure}[ht!]
\begin{center}
\scalebox{0.66}{\includegraphics{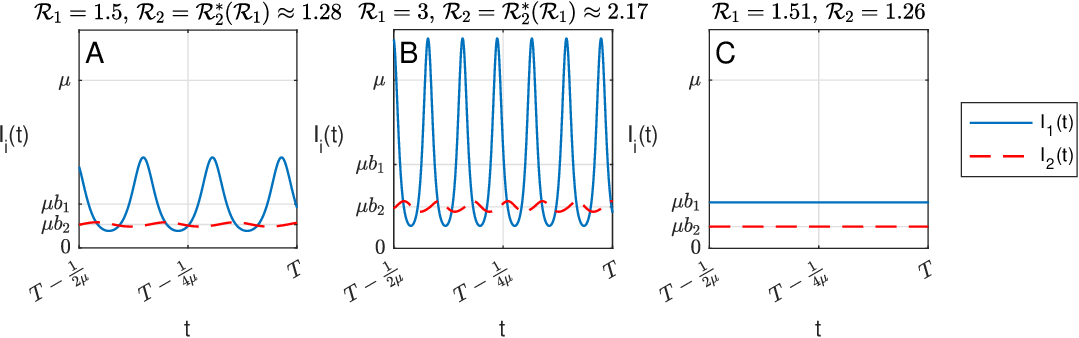}}
\caption{Solutions~$I_1(t)$ and~$I_2(t)$ of~\eqref{eq:model} in the time interval~$[T-1/2\mu,T]$ where~$T=10^{5}$, for the case~$\mu=2\cdot 10^{-4}$,~$\sigma_{12}=1$,~$\sigma_{21}=0.5$,~$\gamma_1=0.9$,~$\gamma_2=1.2$ and A: $\mathcal{R}_1=1.5$,~$\mathcal{R}_2=\mathcal{R}_2^*(\mathcal{R}_1)\approx 1.28$, 
B: $\mathcal{R}_1=3$,~$\mathcal{R}_2=\mathcal{R}_2^*(\mathcal{R}_1)\approx 2.17$, C: $\mathcal{R}_1=1.51$,~$\mathcal{R}_2=1.26$. Note that the y-axis ticks present approximations~$I_i^*\approx \mu b_i+\mathcal{O}(\mu^2)$, see Proposition~\ref{prop:CE_small_mu}.
}
\label{fig:Iit_sigma05}
\end{center}
\end{figure}Figures~\ref{fig:Iit_sigma05}A and~\ref{fig:Iit_sigma05}B present a solution of~\eqref{eq:model} for a set of parameters chosen along two points in~$\Gamma^*$: 
The point with~$\mathcal{R}_1=1.5$ in Figure~\ref{fig:Iit_sigma05}A, and the point with~$\mathcal{R}_1=3$ in Figure~\ref{fig:Iit_sigma05}B.  In particular, we chose~$\mu=2\cdot 10^{-4}$ which corresponds for a unit time of a week to a demographic turnover rate of 75 years.
In both cases, Theorem~\ref{theo:main_result} implies that a steady state of coexistence exists and that it is linearly unstable.  As expected, we observe that the solution approaches a periodic function. 
We have conjectured that the coexistence steady-state if stable when~$(\mathcal{R}_1,\mathcal{R}_2)$ is not near~$\Gamma^*$.  In Figure~\ref{fig:Iit_sigma05}C we perturb the values of~$\mathcal{R}_1$ and~$\mathcal{R}_2$ used in Figure~\ref{fig:Iit_sigma05}A by~$0.5\%$ and~$1.5\%$, respectively.   As expected, we observe that the system converges to the coexistence steady-state.  

\begin{figure}[ht!]
\begin{center}
\scalebox{0.66}{\includegraphics{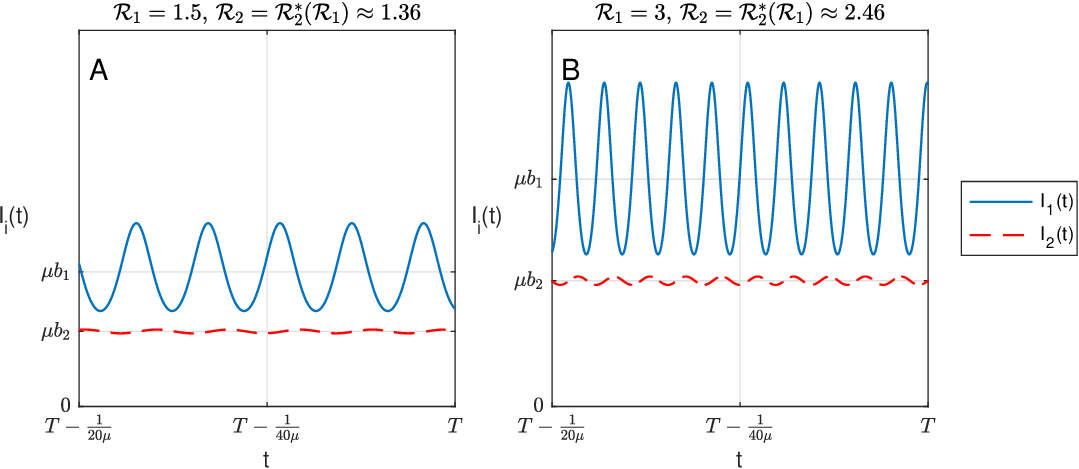}}
\caption{Weak cross-immunity and weak asymmetry: Solutions~$I_1(t)$ and~$I_2(t)$ of~\eqref{eq:model} in the time interval~$[T-\mu/20,T]$ where~$T=10^{7}$, for the case~$\mu= 10^{-6}$,~$\sigma_{12}=1$,~$\sigma_{21}=0.95$,~$\gamma_1=0.9$,~$\gamma_2=1.2$ and A: $\mathcal{R}_1=1.5$,~$\mathcal{R}_2=\mathcal{R}_2^*(\mathcal{R}_1)\approx 1.36$, 
B: $\mathcal{R}_1=3$,~$\mathcal{R}_2=\mathcal{R}_2^*(\mathcal{R}_1)\approx 2.46$.  Note that the y-axis ticks present approximations~$I_i^*\approx \mu b_i+\mathcal{O}(\mu^2)$, see Proposition~\ref{prop:CE_small_mu}.
}
\label{fig:Iit_sigma95}
\end{center}
\end{figure}Theorem~\ref{theo:main_result} points to an oscillatory region in~$(\mathcal{R}_1,\mathcal{R}_2)$ for any value of~$0<\sigma_{21}<1$ and for~$\sigma_{12}\approx1$.  In particular,~$\sigma_{21}$ can be arbitrarily close to~$\sigma_{12}$.  Namely, oscillations are possible for arbitrarily weak cross-immunity.  We now demonstrate this result numerically.  To do so, we consider the case~$\sigma_{21}=0.95$ and~$\sigma_{12}=1$, and sample two points on~$\Gamma^*$, one with~$\mathcal{R}_1=1.5$ and the second with~$\mathcal{R}_1=3$.  When taking~$\mu=2\cdot 10^{-4}$ as in Figure~\ref{fig:Iit_sigma05}, we observe that the system converges to the coexistence steady-state, i.e.,~$\phi_{CE}$ is stable (data not shown).  Yet, as expected, when considering a smaller value of~$\mu$, we observe that the coexistence steady-state is unstable and the system approaches a periodic function, see Figure~\ref{fig:Iit_sigma95}.

The above examples were provided with~$\sigma_{12}=1$.  Theorem~\ref{theo:main_result} also applied to~$\sigma_{12}\approx1$.  This result aligns with numerical observations of periodic solutions of~\eqref{eq:model} with co-circulating strains for~$\mu=0.0004$ and~$\sigma_{12}=0.999386\approx1-1.5\mu$, see~\cite[Figure 5]{chung2016dynamics}.

\section{Beyond~$\Gamma^*$ - A uniform expansion}\label{sec:uniformExpansion}
Theorem~\ref{theo:main_result} describes the behavior of system~\eqref{eq:model} along the curve~$\Gamma^*$ in the parameter plane~$(\mathcal{R}_1,\mathcal{R}_2)$.
In this section, we aim to go beyond~$\Gamma^*$ and provide a uniform approximation of the eigenvalues of the relevant Jacobian in all the feasible parameter space.  
That is, our objective is to derive an approximation of~$\{\lambda_j\}_{j=4}^7$, \eqref{eq:lambda_47}, which, for a sufficiently small~$\mu$, is valid for all~$(\mathcal{R}_1,\mathcal{R}_2)\in\Omega^\mu$,~\eqref{eq:Omega}.

\begin{figure}[ht!]
\begin{center}
\scalebox{0.66}{\includegraphics{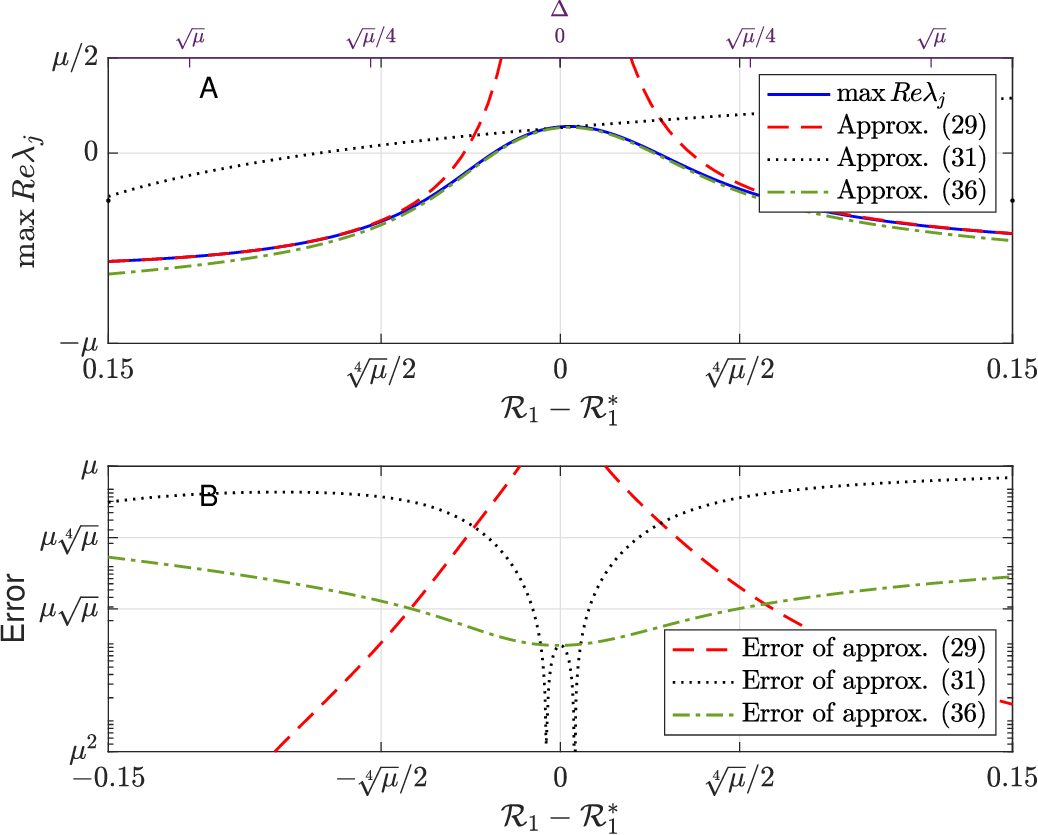}}
\caption{A: Profile of~$\max_{j=4}^7{\rm Re}\lambda_j$ as a function of~$\mathcal{R}_1-\mathcal{R}_1^*$ for the same parameters used in Figure~\ref{fig:Iit_sigma95}, and in particular for~$\mathcal{R}_2=\mathcal{R}_2^*$ (blue solid curve).  Super-imposed are the approximations~$\max_{j=4}^7{\rm Re}\lambda_j^{\rm approx}$ where~$\lambda_j^{\rm approx}$ is given by approximation~\eqref{eq:expansion_with_qi} (dashed red curve), 
approximation~\eqref{eq:expansion_with_si} (dotted black curve), and 
approximation~\eqref{eq:gen_approximation_sj_including_coeffs} (dash-dotted green curve).    B: Same data as in A with corresponding error graphs~$|\max_{j=4}^7{\rm Re}\lambda_j-\max_{j=4}^7{\rm Re}\lambda_j^{\rm approx}|$.
}
\label{fig:uniformExpansion}
\end{center}
\end{figure}The analysis in Section~\ref{sec:CE_perturbation} showed that the approximation~\eqref{eq:expansion_with_qi} for the eigenvalues~$\{\lambda_j\}_{j=4}^7$ of~$P(\lambda)$ is not uniformly valid for all~$(\mathcal{R}_1,\mathcal{R}_2)\in\Omega^\mu$.
Instead, its applicability is limited to a certain distance from~$\Gamma^*$ where~$\Delta=0$, see the red dashed curve in Figure~\ref{fig:uniformExpansion}.  The expansion in~\eqref{eq:expansion_with_si}, on the other hand, is valid for all~$(\mathcal{R}_1,\mathcal{R}_2)\in\Gamma^*$ or when~$\Delta=0$, see the black dotted curve in Figure~\ref{fig:uniformExpansion}.  Consequently, the validity regions of the two approximations are distinct and do not intersect. Our objective is to formulate a more precise approximation for~$\lambda_j$ near~$\Gamma^*$, in order to create an overlapping validity region with~\eqref{eq:expansion_with_qi}, and therefore ensure that the two approximations span the entire domain.   

We first better characterize the validity region of the approximation~\eqref{eq:expansion_with_qi} in the vicinity of~$\Gamma^*$, i.e., when~$0<\Delta\ll1$.  The expansion~\eqref{eq:expansion_with_qi} is valid when
\[
 \mu |q_j|=\frac{\mu |c^q_0|}{16\gamma_1^3\gamma_2^3|r_j^2|\sqrt{\Delta}}\ll \sqrt{\mu}|r_j|,\qquad 0<\Delta\ll1,
\]
where the equality is due to~\eqref{eq:qi}.
Thus, expansion~\eqref{eq:expansion_with_qi} is valid when~$\Delta\gg \mu$.
Subsequently, we refine the approximation~\eqref{eq:expansion_with_si} to ensure its validity when
\[
\Delta \le \sqrt{\mu},
\]
allowing both approximations to overlap near~$\Gamma^*$, where
\[
\mu\ll\Delta\le\sqrt{\mu}.
\]

Let us consider~$\Delta=\delta\sqrt{\mu}$, and the following generalization of expansion~\eqref{eq:expansion_with_si} for~$\lambda_j(\delta)$ 
\begin{subequations}\label{eq:gen_approximation_sj_including_coeffs}
\begin{equation}\label{eq:gen_expansion_with_sj}
\lambda_j(\delta)=\sqrt\mu r_j+\mu^{3/4}s_j(\delta)+\mu q_j^s(\delta)+\mathcal{O}(\mu^{5/4}).
\end{equation}
Substituting~\eqref{eq:gen_expansion_with_sj} in~$P(\lambda)$ and equating by powers of~$\mu$, yields at leading order, Equation~\eqref{eq:eq4ri} for~$r_j$, and at the next order
\begin{equation}
s_j^2(\delta)=s_j^2(0)-\frac{\delta}{8\gamma_1\gamma_2P_{10}},
\end{equation}
where~$s_j(0)$ is given by~\eqref{eq:si}.  To derive a uniform expansion with an error smaller than~$\mathcal{O}(\mu)$, we continue to resolve the~$\mathcal{O}(\mu)$ term in~\eqref{eq:gen_expansion_with_sj}.
The next order yields
\begin{equation}
q_j^s(\delta)=q_j^s(0)-i \frac{\sqrt{2\gamma_1\gamma_2P10}}{8\gamma_1\gamma_2P_{10}^2}\delta,\quad q_j^s(0)=\frac{4\gamma_1^3\gamma_2^3P_{30} - \gamma_1\gamma_2P_{10}^2P_{11}+ (\gamma_1+\gamma_2)P_{10}^3}{4\gamma_1^2\gamma_2^2P_{10}^2}.
\end{equation}
\end{subequations}
Approximation~\eqref{eq:gen_approximation_sj_including_coeffs} of~$\lambda_j$ is valid for~$\Delta=\mathcal{O}(\sqrt\mu)$, see dash-dotted curve in Figure~\ref{fig:uniformExpansion}.

The validity regions of the approximations~\eqref{eq:expansion_with_qi} and~\eqref{eq:gen_approximation_sj_including_coeffs} fully cover the feasible space~\eqref{eq:Omega} and overlap in the vicinity of~$\Gamma^*$ where~$\mu\ll\Delta\le \sqrt{\mu}$.  We now derive a uniform approximation by combining the two approximations.
\begin{theorem}[Uniform approximation]\label{theo:uniform}
Let~$\gamma_1,\gamma_2>0$,~$0<\sigma_{21}<1$ and~$\sigma_{12}=1-s_{12}\mu$,~$s_{12}\ge0$, and define~$\Omega^\mu\subset\{\mathcal{R}_1\ge1,\mathcal{R}_2\ge1\}$ by~\eqref{eq:Omega}.

Then, there exists a sufficiently small~$\mu>0$, such that for any choice of parameters~$(\mathcal{R}_1,\mathcal{R}_2)\in\Omega^\mu$,~\eqref{eq:Omega},
\begin{enumerate}
\item There exists a unique steady-state~$\phi_{CE}=\phi_{CE}(\mathcal{R}_1,\mathcal{R}_2)$ of~\eqref{eq:model} with~$I_1>0$ and~$I_2>0$.  
\item The eigenvalues of the Jacobian of~\eqref{eq:model} at~$\phi_{CE}(\mathcal{R}_1,\mathcal{R}_2)$ satisfy~\eqref{eq:lambda_123}, and for~$j=4,5,\cdots,7$,
\begin{equation}\label{eq:uniformExpansion}
\begin{cases}
\lambda_j=\sqrt\mu\, r_j+\mu q_j+\mathcal{O}(\mu\sqrt\mu)& \Delta \ge \sqrt\mu,\\
\lambda_j=\sqrt\mu\, r_j+\mu^{3/4}s_j(\delta)+\mu q^s_j(\delta)+\mathcal{O}(\mu^{5/4})& \Delta< \sqrt\mu.
\end{cases}
\end{equation}
\end{enumerate}
where~$\Delta=\Delta(\mathcal{R}_1,\mathcal{R}_2)$ is defined by~\eqref{eq:Delta}, and
\[
\delta:=\frac{\Delta}{\sqrt{\mu}},
\]
and the coefficients are defined as follows:~$\{r_j\}_{j=4}^7$ are the roots of~\eqref{eq:eq4ri},~$q_j$ is given by~\eqref{eq:qi}, and~$s_j(\delta),\,q^s_j(\delta)$ are given by~\eqref{eq:gen_approximation_sj_including_coeffs}.
\end{theorem}
\begin{remark}
We note that the analysis of Section~\ref{sec:CE_perturbation} yields, in the case~$\Delta=0$, an~$\mathcal{O}(\mu)$ approximation of~$\{\lambda_j\}_{j=4}^7$ to determine the stability of the steady state of coexistence.  In the current analysis, we continue to resolve higher-order terms in~\eqref{eq:expansion_with_si} to allow a uniform approximation of~$\{\lambda_j\}_{j=4}^7$ with~$o(\mu)$ error.  
Appendix~\ref{app:verificationLambda} presents a numerical verification of the refined approximation~\eqref{eq:gen_expansion_with_sj} of the eigenvalues~$\{\lambda_j\}_{j=4}^7$ in the case~$\Delta=0$.
\end{remark}

\begin{figure}[ht!]
\begin{center}
\scalebox{0.66}{\includegraphics{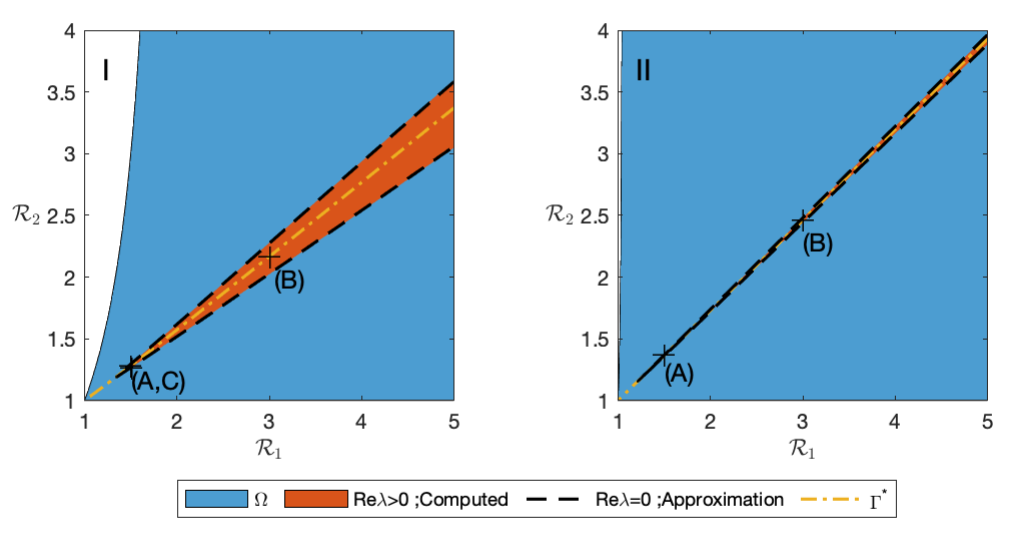}}
\caption{Feasible parameter region~$\Omega^\mu$ (blue area) and region in which the coexistence steady-state is unstable (red area).  Super-imposed are the approximated boundaries of the region in which the coexistence steady-state is unstable (dashed black curve) and the curve~$\Gamma^*$ (dash-dotted yellow curve).
I: Same parameters as those used in Figure~\ref{fig:Iit_sigma05}, expect~$R_i$ values.  Points (A)-(C) correspond to~$\mathcal{R}_i$ values used in graphs~\ref{fig:Iit_sigma05}A-\ref{fig:Iit_sigma05}C, respectively.
II: Same parameters as those used in Figure~\ref{fig:Iit_sigma95}, expect~$R_i$ values.  Points (A), (B) correspond to~$\mathcal{R}_i$ values used in graphs~\ref{fig:Iit_sigma95}A and~\ref{fig:Iit_sigma95}B, respectively.  Details of the numerical methods used to produce these graphs are provided in Appendix~\ref{app:numerical_methods}.
}
\label{fig:bifurcationGraph}
\end{center}
\end{figure}
Theorem~\ref{theo:uniform} presents a uniform approximation~\eqref{eq:uniformExpansion} for~$\lambda_j$, the eigenvalues of the associated Jacobian, across the entire feasible parameter space~$\Omega^\mu$.
Figure~\ref{fig:bifurcationGraph}I illustrates the utility of Theorem~\ref{theo:uniform}, showing~$\Omega^\mu$ (blue area) and the region where~$\max{\rm Re} \lambda_j>0$ (red area), using the same parameters as Figure~\ref{fig:Iit_sigma05}, specifically with~$\mu=2\cdot 10^{-4}$. This red area denotes where the coexistence steady-state~$\phi_{CE}$ is unstable. As anticipated, the curve~$\Gamma^*$ (dash-dotted yellow line) falls within this area. For every point in~$\Omega$, we use~\eqref{eq:uniformExpansion} to estimate~$\max{\rm Re} \lambda_j$, with a black dashed line outlining the contour where~$\max{\rm Re}\lambda^{\rm approx}=0$. The approximation~\eqref{eq:uniformExpansion} closely matches the region of instability for~$\phi_{CE}$. In Figure~\ref{fig:bifurcationGraph}II, we depict the same data with the parameters of Figure~\ref{fig:Iit_sigma95}, again observing that~\eqref{eq:uniformExpansion} accurately captures the instability region for~$\phi_{CE}$.
\section{Conclusions}\label{sec:conclusions}
This study investigates the existence, uniqueness and local stability of a coexistence steady state in a two-strain transmission model featuring partial cross-immunity. We particularly address a long-standing question for this model regarding the conditions required for the existence of periodic solutions.

The key result of this study, Theorem~\ref{theo:main_result}, reveals a new oscillatory regime of the system~\eqref{eq:model} by showing that when one of the strains does not effectively provide cross-immunity to infection by the other strain ($\sigma_{12}\approx1$) and the second strain confers some degree of partial cross-immunity ($0<\sigma_{21}<1$), then a coexistence steady-state uniquely exists in a subset~$\Omega^\mu$ of the parameter plane~$(\mathcal{R}_1,\mathcal{R}_2)$, and it is linearly unstable for values of~$(\mathcal{R}_1,\mathcal{R}_2)$ in a region about a curve~$\Gamma^*\subset(\mathcal{R}_1,\mathcal{R}_2)$. The numerical results show that, as expected, for parameters in the vicinity of~$\Gamma^*$, the system converges to periodic solutions with co-circulating strains.  

The current literature considering sustained oscillation in epidemic systems with cross-immunity~\cite{andreasen1997dynamics,castillo1989epidemiological,Gupta94}, and in particular the work~\cite{chung2016dynamics} that addresses a model identical to ours, implies that the necessary conditions for oscillations include: 1) Sufficiently strong cross-immunity induced by one of the strains.
2) Sufficiently strong asymmetry in cross-immunity parameters. 
The first condition is intuitive since it is known that a multi-strain SIR system does not give rise to oscillations in the absence of cross-immunity. 
In contrast to this picture, our work shows that, surprisingly, sustained oscillations may also arise in cases where cross-immunity is weak and the asymmetry in cross-immunity parameters is weak.  
We demonstrate this result by providing an example of oscillations about an unstable coexistence steady state with $\sigma_{12}=1$ (no cross-immunity) and~$\sigma_{21}=0.95$ (weak cross-immunity), see Figure~\ref{fig:Iit_sigma95}. 
These results are counterintuitive and contradict the results of~\cite[Theorem 1.2]{chung2016dynamics}.  

The contradiction between the results of~\cite[Theorem 1.2]{chung2016dynamics} and our results does not originate from differences in the problem studied.
In fact, the analysis in both studies, namely in~\cite{chung2016dynamics} and the current study, considers the same mathematical model in the same parameter space and focuses on the same asymptotic regime~\eqref{eq:asymptotic_regime}.  Rather, the reason for the contradiction between Theorem~\ref{theo:main_result} and~\cite[Theorem 1.2]{chung2016dynamics} is that the proof of Theorem 1.2 in~\cite{chung2016dynamics} assumes expansion~\eqref{eq:expansion_with_qi} to be uniformly valid, that is, it assumes that for~$0<\mu\ll1$ the relevant eigenvalues behave as
\[
\lambda_j=\sqrt{\mu} r_j+\mu q_j+\mathcal{O}(\mu\sqrt\mu).
\]  
However, our results and, in particular, Lemma~\ref{lem:B1equalsB2} and the proof of Theorem~\ref{theo:main_result}, show that~$q_j\sim 1/\sqrt\Delta$ for small enough~$
0<\mu$.  Hence, expansion~\eqref{eq:expansion_with_qi} becomes invalid in feasible regions of the parameter space, specifically along the curve~$\Gamma^*$ where~$\Delta=0$.  Moreover, Theorem~\ref{theo:main_result} directly associates instabilities of~$\phi_{CE}$ with the cases in which expansion~\eqref{eq:expansion_with_qi} is invalid.  

The insight that expansion~\eqref{eq:expansion_with_qi} is not uniformly valid, the characterization of the region of invalidity, and the association of this region with instabilities became possible due to a combination of somewhat technical results that allowed us to confront the complicated expressions associated with the coexistence steady-state~$\phi^{CE}$, and the characteristic polynomial~\eqref{eq:Plambda}.  These results include: 1) The reduction of the nonlinear algebraic system of eight equations for~$\phi^{CE}$ to a cubic equation, see Proposition~\ref{prop:reduceSystem}. 
2) Lemma~\ref{lem:1minuss_0_lower_bound_gnrl_mu} that defines the auxiliary expression~$A_i$,~$i=1,2$, and proves that they are strictly positive.
3) Expressions~\eqref{eq:P10_P20} for key coefficients of the characteristic polynomial in terms of auxiliary expressions such as~$A_i$.

Theorem~\ref{theo:uniform} provides a uniform approximation of the eigenvalues of the Jacobian at~$\phi_{CE}$ in the feasible parameter space and at an accuracy sufficient to determine the stability of~$\phi_{CE}$.  This result seems to open the way to a study of the stability of~$\phi_{CE}$ outside of~$\Gamma^*$.  Yet, the resulting expressions are long and cumbersome, and do not lend themselves to investigation except in carefully selected cases. Based on numerical observations such as those presented in Figure~\ref{fig:bifurcationGraph}, we conjecture that the steady state of coexistence is linearly stable outside the oscillatory region characterized in this work.

This research targets the scenario~$\sigma_{21}>0$, and our findings hold as~$\sigma_{21}$ approaches zero. To directly address the case where~$\sigma_{21}=0$, adjustments such as omitting~\eqref{eq:Ri} are necessary.  A detailed analysis of this case will be published elsewhere.

The study~\cite{gavish2024revisiting} considers the system~\eqref{eq:model} in the case where one strain is significantly more transmissible than the other,~$\mathcal{R}_1\gg\mathcal{R}_2$, and shows that under broad conditions, the system converges to an endemic steady state of coexistence.  Thus, the competitive exclusion principle does not unconditionally hold beyond the established case of complete immunity, and that in such a case the co-circulation of the strains is not oscillatory but endemic.  The study~\cite{gavish2024revisiting} leaves open the question whether the high transmissibility of both strains also ensures convergence to an endemic steady state of coexistence.  The current work provides an answer to this question by showing that, since~$\Gamma^*$ is unbounded, even when~$\mathcal{R}_1\gg1$ and~$\mathcal{R}_2\gg1$, the steady-state of coexistence is linearly unstable in certain regions of the parameter space. 

Our findings challenge the current understanding of mechanisms for self-oscillations in epidemiological systems, as they point to a narrow oscillatory region in an unexpected parameter domain.   We focus on scenarios where the mortality rate, $\mu$, is much lower than the recovery rates, $\mu \ll \gamma_i$.  The question of whether our findings extend to higher values of $\mu$ arises. For example, could there be an oscillatory region around $\Gamma^*$ when $\mu = \mathcal{O}(\gamma_i)$? Will this region remain narrow near $\Gamma^*$, or become broader? Our numerical analysis indicates that higher values of $\mu$ do not produce oscillations, suggesting that our results are genuinely confined to~$\mu \ll \gamma_i$. Somehow similar behaviors are observed in enzymatic systems with substrate inhibition kinetics~\cite{shen1994role}, where the oscillatory region is so narrow that it cannot be detected experimentally. The relevance of our results to real biological systems remains unresolved.  A systematic study of the oscillatory mechanism and its possible relation to the oscillations in~\cite{shen1994role} will be published elsewhere.

This study began as an effort to extend the analysis in~\cite{chung2016dynamics} to a broader epidemic model. We discovered an oversight in~\cite{chung2016dynamics}, and its correction revealed a new mechanism for self-oscillations. This mechanism may also apply to epidemic models considering additional factors such as changes in infectivity, quarantine, or age structure. Further examination of these models, potentially uncovering new oscillatory phenomena or providing new insights into known results, will be presented elsewhere.
\backmatter

\bmhead{Acknowledgements}
This research was supported by the Israel Science Foundation (grant no. 3730/20) within the KillCorona-Curbing Coronavirus Research Program, and by the Israeli Science Foundation (ISF) grant 1596/23.

The author is grateful to Guy Katriel for his valuable comments and discussions on this research.




\begin{appendices}

\section{Numerical verification - stability analysis}\label{app:verificationLambda}

\begin{figure}[ht]\centering
	\includegraphics[width=0.7\textwidth]{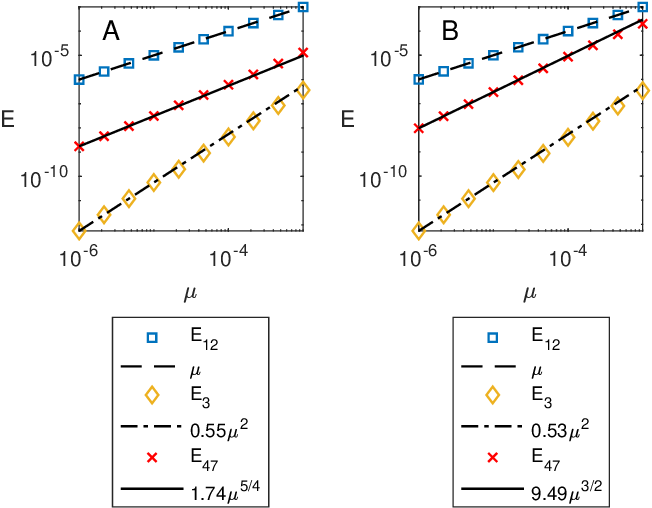}
	\caption{Graph of approximation error~$E_{1,2}=\max\{E_1,E_2\}$ ($\square$),~$E_3$ ($\diamond$) and~$E_{4-7}=\max\{E_4,\cdots,E_7\}$ ($\times$) as a function of~$\mu$ where~$E_i$ is given by~\eqref{eq:Ei} with the parameters used in: A)~Figure~\ref{fig:Iit_sigma05}A.  Super-imposed are the curves~$\mu$ (dashes),~$0.55\mu^2$ (dash-dots) and~$1.74\,\tau\sqrt[4]{\tau}$ (solid).  B)~Figure~\ref{fig:Iit_sigma05}C.  Super-imposed are the curves~$\mu$ (dashes),~$0.53\mu^2$ (dash-dots) and~$9.49\,\tau\sqrt{\tau}$ (solid). \label{fig:err_as_func_of_mu}
 }
\end{figure}Figure~\ref{fig:err_as_func_of_mu}A 
presents the approximation errors
\begin{equation}\label{eq:Ei}
E_i=|\lambda_i-\lambda_i^{\rm approx}|,
\end{equation}
where~$\lambda_i^{\rm approx}$ is the approximation~\eqref{eq:expansion_with_si} to the root~$\lambda_i$ of~\eqref{eq:Plambda} for the set of parameters used in Figure~\ref{fig:Iit_sigma05}A but for a range of values of~$\mu$.  As expected, the error~$E_{1,2}=\max_{i=1,2}E_i=\mathcal{O}(\mu)$ and~$E_{3}=\mathcal{O}(\mu^2)$, see~\eqref{eq:lambda_123}. 
The set of parameters used in Figure~\ref{fig:Iit_sigma05} corresponds to a case~$\Delta=0$, thus we expect that~$\{\lambda_i\}_{i=4}^7$ behaves as~\eqref{eq:expansion_with_si} with an~$\mathcal{O}(\mu\sqrt[4]{\mu})$ error.  Indeed, we observe that~$E_{4-7}=\max\{E_4,\cdots,E_7\} \sim 1.74\mu\sqrt[4]{\mu}$.  

Figure~\ref{fig:err_as_func_of_mu}B 
presents the approximation errors~\eqref{eq:Ei} for the set of parameters used in Figure~\ref{fig:Iit_sigma05}C but for a range of values of~$\mu$.  In this case,~$\Delta\approx 0.002$, so that we expect~$\{\lambda_i\}_{i=4}^7$ to behave according to~\eqref{eq:expansion_with_qi} with an~$\mathcal{O}(\mu\sqrt{\mu})$ error.  Indeed, we observe that~$E_{4-7}=\max\{E_4,\cdots,E_7\}\sim 9.49\mu\sqrt{\mu}$. 

\section{Proof of Theorem~\ref{theo:main_result}}\label{app:proof_main_result}
The proof of Theorem~\ref{theo:main_result} relies only on the leading order approximation~$s_0$ of~$S^*$, and on the coefficients~$P_{10}, P_{11}, P_{20}$ and~$P_{30}$ of the characteristic polynomial~\eqref{eq:Plambda}.   We now show that these quantities are not effects by perturbing~$\sigma_{12}=1-\mu s_{12}$.
Indeed, an~$\mathcal{O}(\mu)$ perturbation from~$\sigma_{12}=1$ does not effect the leading order approximation~$s_0$ of~$S^*$, but rather its correction terms.  Similarly, an~$\mathcal{O}(\mu)$ perturbation from~$\sigma_{12}=1$, results in a change to~$P_{22},P_{23},P_{31},\cdots$,
\[
P(\lambda;\sigma_{12}=1-s_{12}\mu)-P(\lambda;\sigma_{12}=1)=\mu^2\lambda^3 (P_{22}+\lambda P_{23}+\cdots)+\mu^3\lambda(P_{31}+\lambda P_{32}+\cdots)+\mathcal{O}(\mu^4).
\]

\section{Numerical methods}\label{app:numerical_methods}
We use the standard Matlab {\tt ode45} command to compute the solution of~\eqref{eq:model}, e.g., to produce the graphs of Figures~\ref{fig:Iit_sigma05} and~\ref{fig:Iit_sigma05}.

Figure~\ref{fig:bifurcationGraph} shows the bifurcation graph, which necessitates calculating~$\phi_{CE}$ using certain parameters, followed by determining the eigenvalues of the Jacobian from~\eqref{eq:model} around~$\phi_{CE}$. To find the compartment sizes of~$\phi_{CE}$, we numerically solve the steady-state nonlinear algebraic system \eqref{eq:model} with Matlab's {\tt fsolve} command. The Jacobian analytic form is provided and the leading-order approximation~\eqref{eq:approxS*_andIi*} is used as an initial estimate. The eigenvalues of the Jacobian at~$\phi_{CE}$ are computed using Matlab's standard {\tt eig} function.
\end{appendices}


\begin{thebibliography}{19}
\ifx \bisbn   \undefined \def \bisbn  #1{ISBN #1}\fi
\ifx \binits  \undefined \def \binits#1{#1}\fi
\ifx \bauthor  \undefined \def \bauthor#1{#1}\fi
\ifx \batitle  \undefined \def \batitle#1{#1}\fi
\ifx \bjtitle  \undefined \def \bjtitle#1{#1}\fi
\ifx \bvolume  \undefined \def \bvolume#1{\textbf{#1}}\fi
\ifx \byear  \undefined \def \byear#1{#1}\fi
\ifx \bissue  \undefined \def \bissue#1{#1}\fi
\ifx \bfpage  \undefined \def \bfpage#1{#1}\fi
\ifx \blpage  \undefined \def \blpage #1{#1}\fi
\ifx \burl  \undefined \def \burl#1{\textsf{#1}}\fi
\ifx \doiurl  \undefined \def \doiurl#1{\url{https://doi.org/#1}}\fi
\ifx \betal  \undefined \def \betal{\textit{et al.}}\fi
\ifx \binstitute  \undefined \def \binstitute#1{#1}\fi
\ifx \binstitutionaled  \undefined \def \binstitutionaled#1{#1}\fi
\ifx \bctitle  \undefined \def \bctitle#1{#1}\fi
\ifx \beditor  \undefined \def \beditor#1{#1}\fi
\ifx \bpublisher  \undefined \def \bpublisher#1{#1}\fi
\ifx \bbtitle  \undefined \def \bbtitle#1{#1}\fi
\ifx \bedition  \undefined \def \bedition#1{#1}\fi
\ifx \bseriesno  \undefined \def \bseriesno#1{#1}\fi
\ifx \blocation  \undefined \def \blocation#1{#1}\fi
\ifx \bsertitle  \undefined \def \bsertitle#1{#1}\fi
\ifx \bsnm \undefined \def \bsnm#1{#1}\fi
\ifx \bsuffix \undefined \def \bsuffix#1{#1}\fi
\ifx \bparticle \undefined \def \bparticle#1{#1}\fi
\ifx \barticle \undefined \def \barticle#1{#1}\fi
\bibcommenthead
\ifx \bconfdate \undefined \def \bconfdate #1{#1}\fi
\ifx \botherref \undefined \def \botherref #1{#1}\fi
\ifx \url \undefined \def \url#1{\textsf{#1}}\fi
\ifx \bchapter \undefined \def \bchapter#1{#1}\fi
\ifx \bbook \undefined \def \bbook#1{#1}\fi
\ifx \bcomment \undefined \def \bcomment#1{#1}\fi
\ifx \oauthor \undefined \def \oauthor#1{#1}\fi
\ifx \citeauthoryear \undefined \def \citeauthoryear#1{#1}\fi
\ifx \endbibitem  \undefined \def \endbibitem {}\fi
\ifx \bconflocation  \undefined \def \bconflocation#1{#1}\fi
\ifx \arxivurl  \undefined \def \arxivurl#1{\textsf{#1}}\fi
\csname PreBibitemsHook\endcsname

\bibitem[\protect\citeauthoryear{Martcheva}{2015}]{martcheva2015introduction}
\begin{bbook}
\bauthor{\bsnm{Martcheva}, \binits{M.}}:
\bbtitle{An Introduction to Mathematical Epidemiology}
vol. \bseriesno{61}.
\bpublisher{Springer},
\blocation{New York}
(\byear{2015})
\end{bbook}
\endbibitem

\bibitem[\protect\citeauthoryear{Wormser and
  Pourbohloul}{2008}]{wormser2008modeling}
\begin{botherref}
\oauthor{\bsnm{Wormser}, \binits{G.P.}},
\oauthor{\bsnm{Pourbohloul}, \binits{B.}}:
Modeling Infectious Diseases in Humans and Animals By Matthew James Keeling and
  Pejman Rohani Princeton, NJ: Princeton University Press, 2008. 408 pp.,
  Illustrated (hardcover).
The University of Chicago Press
(2008)
\end{botherref}
\endbibitem

\bibitem[\protect\citeauthoryear{Andreasen
  et~al.}{1997}]{andreasen1997dynamics}
\begin{barticle}
\bauthor{\bsnm{Andreasen}, \binits{V.}},
\bauthor{\bsnm{Lin}, \binits{J.}},
\bauthor{\bsnm{Levin}, \binits{S.A.}}:
\batitle{The dynamics of cocirculating influenza strains conferring partial
  cross-immunity}.
\bjtitle{Journal of mathematical biology}
\bvolume{35},
\bfpage{825}--\blpage{842}
(\byear{1997})
\end{barticle}
\endbibitem

\bibitem[\protect\citeauthoryear{Gomes and Medley}{2002}]{gomes2002dynamics}
\begin{bchapter}
\bauthor{\bsnm{Gomes}, \binits{M.G.M.}},
\bauthor{\bsnm{Medley}, \binits{G.F.}}:
\bctitle{Dynamics of multiple strains of infectious agents coupled by
  cross-immunity: a comparison of models}.
In: \bbtitle{Mathematical Approaches for Emerging and Reemerging Infectious
  Diseases: Models, Methods, and Theory},
pp. \bfpage{171}--\blpage{191}
(\byear{2002}).
\bcomment{Springer}
\end{bchapter}
\endbibitem

\bibitem[\protect\citeauthoryear{Dawes and Gog}{2002}]{dawes2002onset}
\begin{barticle}
\bauthor{\bsnm{Dawes}, \binits{J.}},
\bauthor{\bsnm{Gog}, \binits{J.}}:
\batitle{The onset of oscillatory dynamics in models of multiple disease
  strains}.
\bjtitle{Journal of Mathematical Biology}
\bvolume{45}(\bissue{6}),
\bfpage{471}--\blpage{510}
(\byear{2002})
\end{barticle}
\endbibitem

\bibitem[\protect\citeauthoryear{Lin et~al.}{1999}]{lin1999dynamics}
\begin{barticle}
\bauthor{\bsnm{Lin}, \binits{J.}},
\bauthor{\bsnm{Andreasen}, \binits{V.}},
\bauthor{\bsnm{Levin}, \binits{S.A.}}:
\batitle{Dynamics of influenza a drift: the linear three-strain model}.
\bjtitle{Mathematical biosciences}
\bvolume{162}(\bissue{1-2}),
\bfpage{33}--\blpage{51}
(\byear{1999})
\end{barticle}
\endbibitem

\bibitem[\protect\citeauthoryear{Ferguson et~al.}{1999}]{ferguson1999effect}
\begin{barticle}
\bauthor{\bsnm{Ferguson}, \binits{N.}},
\bauthor{\bsnm{Anderson}, \binits{R.}},
\bauthor{\bsnm{Gupta}, \binits{S.}}:
\batitle{The effect of antibody-dependent enhancement on the transmission
  dynamics and persistence of multiple-strain pathogens}.
\bjtitle{Proceedings of the National Academy of Sciences}
\bvolume{96}(\bissue{2}),
\bfpage{790}--\blpage{794}
(\byear{1999})
\end{barticle}
\endbibitem

\bibitem[\protect\citeauthoryear{Gupta et~al.}{1998}]{gupta1998chaos}
\begin{barticle}
\bauthor{\bsnm{Gupta}, \binits{S.}},
\bauthor{\bsnm{Ferguson}, \binits{N.}},
\bauthor{\bsnm{Anderson}, \binits{R.}}:
\batitle{Chaos, persistence, and evolution of strain structure in antigenically
  diverse infectious agents}.
\bjtitle{Science}
\bvolume{280}(\bissue{5365}),
\bfpage{912}--\blpage{915}
(\byear{1998})
\end{barticle}
\endbibitem

\bibitem[\protect\citeauthoryear{Nu{\~n}o et~al.}{2005}]{nuno2005dynamics}
\begin{barticle}
\bauthor{\bsnm{Nu{\~n}o}, \binits{M.}},
\bauthor{\bsnm{Feng}, \binits{Z.}},
\bauthor{\bsnm{Martcheva}, \binits{M.}},
\bauthor{\bsnm{Castillo-Chavez}, \binits{C.}}:
\batitle{Dynamics of two-strain influenza with isolation and partial
  cross-immunity}.
\bjtitle{SIAM Journal on Applied Mathematics}
\bvolume{65}(\bissue{3}),
\bfpage{964}--\blpage{982}
(\byear{2005})
\end{barticle}
\endbibitem

\bibitem[\protect\citeauthoryear{Nuno et~al.}{2008}]{nuno2008mathematical}
\begin{botherref}
\oauthor{\bsnm{Nuno}, \binits{M.}},
\oauthor{\bsnm{Castillo-Chavez}, \binits{C.}},
\oauthor{\bsnm{Feng}, \binits{Z.}},
\oauthor{\bsnm{Martcheva}, \binits{M.}}:
Mathematical models of influenza: the role of cross-immunity, quarantine and
  age-structure.
Mathematical Epidemiology,
349--364
(2008)
\end{botherref}
\endbibitem

\bibitem[\protect\citeauthoryear{Thieme}{2007}]{thieme2007pathogen}
\begin{botherref}
\oauthor{\bsnm{Thieme}, \binits{H.R.}}:
Pathogen competition and coexistence and the evolution of virulence.
Mathematics for Life Science and Medicine,
123--153
(2007)
\end{botherref}
\endbibitem

\bibitem[\protect\citeauthoryear{Kuddus et~al.}{2022}]{kuddus2022analysis}
\begin{barticle}
\bauthor{\bsnm{Kuddus}, \binits{M.A.}},
\bauthor{\bsnm{McBryde}, \binits{E.S.}},
\bauthor{\bsnm{Adekunle}, \binits{A.I.}},
\bauthor{\bsnm{Meehan}, \binits{M.T.}}:
\batitle{Analysis and simulation of a two-strain disease model with nonlinear
  incidence}.
\bjtitle{Chaos, Solitons \& Fractals}
\bvolume{155},
\bfpage{111637}
(\byear{2022})
\end{barticle}
\endbibitem

\bibitem[\protect\citeauthoryear{Castillo-Chavez
  et~al.}{1989}]{castillo1989epidemiological}
\begin{barticle}
\bauthor{\bsnm{Castillo-Chavez}, \binits{C.}},
\bauthor{\bsnm{Hethcote}, \binits{H.W.}},
\bauthor{\bsnm{Andreasen}, \binits{V.}},
\bauthor{\bsnm{Levin}, \binits{S.A.}},
\bauthor{\bsnm{Liu}, \binits{W.M.}}:
\batitle{Epidemiological models with age structure, proportionate mixing, and
  cross-immunity}.
\bjtitle{Journal of mathematical biology}
\bvolume{27},
\bfpage{233}--\blpage{258}
(\byear{1989})
\end{barticle}
\endbibitem

\bibitem[\protect\citeauthoryear{Chung and Lui}{2016}]{chung2016dynamics}
\begin{barticle}
\bauthor{\bsnm{Chung}, \binits{K.}},
\bauthor{\bsnm{Lui}, \binits{R.}}:
\batitle{Dynamics of two-strain influenza model with cross-immunity and no
  quarantine class}.
\bjtitle{Journal of Mathematical Biology}
\bvolume{73},
\bfpage{1467}--\blpage{1489}
(\byear{2016})
\end{barticle}
\endbibitem

\bibitem[\protect\citeauthoryear{Nemhauser et~al.}{2023}]{CDCyellowbook2024}
\begin{botherref}
\oauthor{\bsnm{Nemhauser}, \binits{J.}},
\oauthor{\bsnm{LaRocque}, \binits{R.}},
\oauthor{\bsnm{Alvarado-Ramy}, \binits{F.}},
\oauthor{\bsnm{Angelo}, \binits{K.}},
\oauthor{\bsnm{Ericsson}, \binits{C.}},
\oauthor{\bsnm{Gertz}, \binits{A.}},
\oauthor{\bsnm{Kozarsky}, \binits{P.}},
\oauthor{\bsnm{Ostroff}, \binits{S.}},
\oauthor{\bsnm{Ryan}, \binits{E.}},
\oauthor{\bsnm{Shlim}, \binits{D.}},
\oauthor{\bsnm{Stauffer}, \binits{W.}},
\oauthor{\bsnm{Weinberg}, \binits{M.}},
\oauthor{\bsnm{Wilson}, \binits{M.E.}},
\oauthor{\bsnm{Keir}, \binits{J.}},
\oauthor{\bsnm{Leidel}, \binits{L.}},
\oauthor{\bsnm{Crowe}, \binits{S.}},
\oauthor{\bsnm{Guinn}, \binits{A.}}:
{CDC} {Y}ellow {B}ook 2024: Health information for international travel.
Oxford University Press
(2023)
\end{botherref}
\endbibitem

\bibitem[\protect\citeauthoryear{Holmes}{2012}]{holmes2012introduction}
\begin{bbook}
\bauthor{\bsnm{Holmes}, \binits{M.H.}}:
\bbtitle{Introduction to Perturbation Methods}
vol. \bseriesno{20}.
\bpublisher{Springer},
\blocation{New York}
(\byear{2012})
\end{bbook}
\endbibitem

\bibitem[\protect\citeauthoryear{Gupta et~al.}{1994}]{Gupta94}
\begin{barticle}
\bauthor{\bsnm{Gupta}, \binits{S.}},
\bauthor{\bsnm{Swinton}, \binits{J.}},
\bauthor{\bsnm{Anderson}, \binits{R.M.}}:
\batitle{Theoretical studies of the effects of heterogeneity in the parasite
  population on the transmission dynamics of malaria}.
\bjtitle{Proceedings of the Royal Society of London. Series B: Biological
  Sciences}
\bvolume{256}(\bissue{1347}),
\bfpage{231}--\blpage{238}
(\byear{1994})
\doiurl{10.1098/rspb.1994.0075}
\end{barticle}
\endbibitem

\bibitem[\protect\citeauthoryear{Gavish}{2024}]{gavish2024revisiting}
\begin{botherref}
\oauthor{\bsnm{Gavish}, \binits{N.}}:
Revisiting the exclusion principle in epidemiology at its ultimate limit.
arXiv preprint arXiv:2405.09813
(2024)
\end{botherref}
\endbibitem

\bibitem[\protect\citeauthoryear{Shen and Larter}{1994}]{shen1994role}
\begin{barticle}
\bauthor{\bsnm{Shen}, \binits{P.}},
\bauthor{\bsnm{Larter}, \binits{R.}}:
\batitle{Role of substrate inhibition kinetics in enzymatic chemical
  oscillations}.
\bjtitle{Biophysical journal}
\bvolume{67}(\bissue{4}),
\bfpage{1414}--\blpage{1428}
(\byear{1994})
\end{barticle}
\endbibitem

\end{thebibliography}

\end{document}